\documentclass[reqno, 11pt]{amsart}

\usepackage[foot]{amsaddr}
\usepackage{fullpage}  
\usepackage{calc,graphicx,amsfonts,amsthm,amscd,epsfig,psfrag,amsmath,amssymb,enumerate,dsfont}
\usepackage{subfig} 
\usepackage{pdfsync}
\usepackage{stmaryrd} 
\usepackage[numeric,initials,nobysame]{amsrefs}

\usepackage[colorlinks=true, pdfstartview=FitV, linkcolor=blue, citecolor=blue, urlcolor=blue,pagebackref=false]{hyperref}
\usepackage{appendix}

\usepackage[showlabels,sections,floats,textmath,displaymath]{}
\newcommand{\ie}{\hbox{\it i.e.\ }}

\reversemarginpar
\newlength\fullwidth
\numberwithin{equation}{section}

\DeclareMathSymbol{\leqslant}{\mathalpha}{AMSa}{"36} 
\DeclareMathSymbol{\geqslant}{\mathalpha}{AMSa}{"3E} 
\DeclareMathSymbol{\eset}{\mathalpha}{AMSb}{"3F}     
\renewcommand{\leq}{\;\leqslant\;}                   
\renewcommand{\geq}{\;\geqslant\;}                   

\renewcommand{\b}{\beta}

\newcommand{\1}{\mathds{1}}

\newcommand{\var}{\operatorname{Var}}

\newcommand{\D}{\Delta}

\renewcommand{\b}{\beta}
\renewcommand{\l}{\lambda}
\renewcommand{\L}{\Lambda}

\renewcommand{\l}{\lambda}
\renewcommand{\a}{\alpha}
\renewcommand{\d}{\delta}
\renewcommand{\t}{\tau}

\newcommand{\g}{\gamma}
\newcommand{\G}{\Gamma}

\newcommand{\e}{\varepsilon}

\renewcommand{\O}{\Omega}

\newcommand{\gap}{{\rm gap}}
\newcommand{\tc}{\thinspace |\thinspace}

\newtheorem{theorem}{Theorem}[section]
\newtheorem{theoremA}{Theorem}
\newtheorem{lemma}[theorem]{Lemma}
\newtheorem{proposition}[theorem]{Proposition}
\newtheorem{corollary}[theorem]{Corollary}
\newtheorem{remark}[theorem]{Remark}

\newtheorem{claim}[theorem]{Claim}
\newtheorem{definition}[theorem]{Definition}

\newcommand{\cA}{\ensuremath{\mathcal A}}

\newcommand{\cD}{\ensuremath{\mathcal D}}
\newcommand{\cE}{\ensuremath{\mathcal E}}

\newcommand{\cG}{\ensuremath{\mathcal G}}

\newcommand{\cL}{\ensuremath{\mathcal L}}

\newcommand{\cR}{\ensuremath{\mathcal R}}

\newcommand{\bbD}{{\ensuremath{\mathbb D}} }
\newcommand{\bbE}{{\ensuremath{\mathbb E}} }

\newcommand{\bbL}{{\ensuremath{\mathbb L}} }

\newcommand{\bbN}{{\ensuremath{\mathbb N}} }

\newcommand{\bbP}{{\ensuremath{\mathbb P}} }

\newcommand{\bbR}{{\ensuremath{\mathbb R}} }

\newcommand{\bbZ}{{\ensuremath{\mathbb Z}} }

\let\a=\alpha \let\b=\beta   \let\d=\delta  \let\e=\varepsilon
 \let\g=\gamma \let\h=\eta      \let\l=\lambda
          \let\p=\pi  
  \let\s=\sigma \let\t=\tau   
  
\let\D=\Delta   \let\G=\Gamma  \let\L=\Lambda 
\let\O=\Omega      

\newcommand{\sst}{{\,:\,}} 
\newcommand{\dyn}[3]{\Phi_{#2,#3}(#1)}
\renewcommand{\div}{{\rm{div}}\,}

\renewcommand{\le}{\leq}

\usepackage[normalem]{ulem}

\begin{document}

\author[P. Chleboun]{P. Chleboun$^1$}
 \address{$^1$Dip. Matematica,  Universit\`a  Roma Tre, Largo S.L.Murialdo 00146, Roma, Italy.}
\email{paul@chleboun.co.uk}

\author[A. Faggionato]{A. Faggionato$^2$}
\address{$^2$Dip. Matematica ``G. Castelnuovo", Universit\`a ``La
  Sapienza''. P.le Aldo Moro  2, 00185  Roma, Italy. }
\email{faggiona@mat.uniroma1.it}

 \author[F. Martinelli]{F. Martinelli$^1$}
\email{martin@mat.uniroma3.it}

\thanks{\sl Work supported by the European Research Council through the ``Advanced
Grant'' PTRELSS 228032.}

\title{Time Scale separation and dynamic heterogeneity  in the low
  temperature East model }

\begin{abstract}
We consider the non-equilibrium
dynamics of the East model, a linear chain of $0$-$1$ spins evolving under a
simple Glauber dynamics in the
presence of a kinetic
constraint which forbids flips of those spins whose left neighbor is $1$. We focus on the glassy
effects caused by the kinetic constraint as
$q\downarrow 0$, where $q$ is the 
equilibrium density of the $0$'s. In the physical literature this limit is equivalent to the zero
temperature limit. We first prove that, for any given $L=O(1/q)$, the divergence as $q\downarrow 0$ of three basic
characteristic time scales of the East process of length $L$ is the
same. Then we examine the problem of dynamic
heterogeneity, i.e. non-trivial spatio-temporal fluctuations of the
local relaxation to equilibrium, one of the central aspects of glassy
dynamics. For any mesoscopic length scale $L=O(q^{-\g})$, $\g<1$, we
show that the characteristic time scale of two East processes of length $L$
and $\l L$ respectively are
indeed separated by a factor $q^{-\a}$, $\a=\alpha(\g)>0$, provided
that $\l\ge 2$ is large enough (independent
of $q$, $\l=2$ for $\g<1/2$). In particular, the
evolution of 
mesoscopic \emph{domains}, {\it i.e.} maximal blocks of the form $111..10$,
occurs on a time scale which depends sharply on the size of the domain, a
clear signature of dynamic heterogeneity. 
A key result for this part is
a very precise computation of the relaxation time of the chain as a
function of $(q,L)$, well
beyond the current knowledge, which uses induction on length scales on one
hand and a novel \emph{algorithmic} lower bound on the other.      
Finally we show that no form
of time scale separation occurs for $\g=1$, {\it
  i.e.} at the equilibrium scale
$L=1/q$, contrary to what was assumed in the physical
literature based on numerical simulations. 
\end{abstract}

\maketitle
\setcounter{tocdepth}{1}
\tableofcontents

\section{Introduction}
Kinetically constrained spin models (KCMs) are interacting
$0$-$1$ particle systems on general graphs evolving with a simple Glauber dynamics which can be
described as follows. At every vertex $x$ the system tries to update the occupancy variable (or spin) at
$x$ to the value $1$ or $0$ with probability $1-q$ and $q$
respectively. However the
update at $x$ is accepted only if a certain local contraint  is verified by
the \emph{current} spin configuration. That explains the wording
``kinetically constrained''. The constraint at site $x$ is always assumed
not to depend on the spin at $x$ and therefore the product
Bernoulli($1-q$) measure $\pi$ is the reversible measure.  Typical
constraints require, for example, that a certain number of the
neighboring spins are in state $0$, or more restrictive, that certain
\emph{preassigned} neighboring spins are in state $0$ (e.g. the
children of $x$ when the underlying graph is a rooted tree). 

The main interest in the physical literature for KCMs (see e.g. 
\cites{Berthier,Ritort,SE2,GarrahanSollichToninelli,SE1,Cancrini:2006uu}) stems from the
fact  that they display many key dynamical
features of real glassy materials: ergodicity breaking
transition at some critical value $q_c$, huge relaxation time for $q$
close to $q_c$, dynamic heterogeneity (non-trivial spatio-temporal fluctuations of the
local relaxation to equilibrium) and aging just to mention a
few. Mathematically, despite of their simple definition, KCMs pose
very challenging and interesting problems because of the hardness of
the constraint, with ramifications towards bootstrap percolation
problems \cite{Spiral}, combinatorics \cites{CDG,Valiant:2004cb},
coalescence processes \cites{FMRT-cmp,FMRT} and  random walks on
triangular matrices \cite{Peres-Sly}. Some of the mathematical tools
developed for the analysis of the relaxation process of KCMs proved to
be quite powerful also in other contexts such as card shuffling problems
\cite{Bhatnagar:2007tr} and random evolution of surfaces \cite{PietroCaputo:2012vl}.

In this paper we analyze one of the most popular KCMs namely, the East model (see e.g. \cites{JACKLE,SE1,SE2} and
\cites{Aldous,CMRT,East-Rassegna}),  which consists of a linear chain (finite or infinite) of $0$-$1$
spins evolving under a simple Glauber dynamics in the  presence of the
kinetic constraint which forbids flips of those spins whose left
neighbor is in state $1$. To avoid trivial irreducibility issues when the chain is either finite or semi-infinite, one
always assumes that the leftmost spin is unconstrained.

It is known that the relaxation time (cf.
\eqref{rilasso}) of the East model is uniformly bounded in the length $L$
of the chain
\cites{Aldous,CMRT} and that, because of the constraints, it
diverges very rapidly as $q\downarrow
0$   (cf. Prop. \ref{prem} below). The mixing time (cf. \eqref{mescolo})
instead diverges linearly in $L$. It is also known \cite{CMST} that,
starting from a large class of initial laws (e.g. a non-trivial Bernoulli$(1-q')$
product measure, $q'\neq q$), the expected value at time $t$ of
a local function $f$ converges exponentially fast to $\pi(f)$, the mean of $f$
w.r.t. the reversible measure $\pi$.  These results prove, in a broad
sense, exponential relaxation to equilibrium for time scales \emph{larger}
than the relaxation time for the infinite chain. 

However the most interesting and
challenging dynamical behavior, featuring aging and dynamic
heterogeneity, occurs for $q\ll1$ on time scales \emph{shorter} than the
relaxation time. Building upon the non-rigorous picture in the
physics literature
\cite{SE1} but going well beyond it, it was
proved in \cite{FMRT-cmp} that, for all $N$ independent of $q$,  the dynamics of the infinite East chain
in a space window $[1,2^N]$ and up to time scales $O(1/q^N)$ is
well approximated, as $q\downarrow 0$, by a certain \emph{hierarchical
  coalescence process} (HCP) \cite{FMRT}. In this HCP vacancies are
isolated and domains (maximal blocks of the form $111..10$) with cardinality
between $2^{n-1}$ and $2^n$, $n\le N$, coalesce with the
domain at their right only at time
scale $\sim (1/q)^n$. As a result, aging and dynamic
heterogeneity in the above regime emerge in a natural way, with a scaling limit for the
relevant quantities in the same universality class as
several other mean field coalescence models of statistical physics
\cite{Derrida}. 

However the  above result says nothing about the
dynamics and its characteristic time scales at
intermediate (mesoscopic) length scales $L=1/q^\g$, $0<\g<1$, or at the typical
inter-vacancy equilibrium scale $L_c=1/q$. As clarified later on (cf.
Section \ref{heuristic}), at these length scales the low temperature dynamics
of the East model is no longer predominantly driven by an effective energy landscape
as in \cite{FMRT-cmp}, but entropic effects become crucial and
a subtle entropy/energy competition comes into play. In the
physical literature these effects have been neglected and the
characteristic time scale $t_n\approx 1/q^n$, appropriate for domains of
length $L_n\approx 2^n$, $n=O(1)$, has been extrapolated up to the
equilibrium scale $L_c=2^{\log_2(1/q)}$
leading for example to the wrong prediction of a relaxation time $\sim
(1/q)^{\log_2(1/q)}$ (to be compared with the bounds given in  Theorem \ref{paletti}).

In this paper we analyse the low temperature dynamics in the above
setting. We first show (cf. Theorem \ref{equo}) that three natural characteristic time scales of
the East model of length $L=O(1/q)$ have a similar scaling as
$q\downarrow 0$. This equivalence is important because of various
notions of ``relaxation time'' which appear in the physical
literature. Secondly we
prove a sharp separation of time scales and dynamic heterogeneity (cf. Theorems \ref{bubu} and \ref{hetero})  
at mesoscopic length scale $L=1/q^\g$, $0<\g<1$.
A key ingredient
for the above results is a novel and detailed computation of
the relaxation time of finite East chains as $q\downarrow 0$,  in which the entropic
contributions in the upper and lower bounds are pinned down very precisely (cf. Theorem
\ref{paletti}). The upper bound is obtained via a substantial
refinement of the inductive technique first introduced in
\cite{CMRT}. Instead the lower bound has been inspired by capacity
methods and is obtained
via a novel construction in the configuration space of a
bottleneck. The equilibrium weight of the bottleneck is
computed in an algorithmic fashion.

In Theorem \ref{sorpresa} we also prove that no time scale separation occurs for $\g=1$, {\it i.e.} at the equilibrium scale
$L_c=1/q$. This  precludes the time scale separation hypothesis in
\cites{SE1,SE2}, put forward on the basis of numerical simulations, which was a keystone of the
\emph{super-domain dynamics} formulation. This
is an example of a case in which numerical simulations can be
misleading because of the extremely long time scales involved,
emphasizing the need for rigorous work.

The above results combined with the hierarchical coalescence picture
of \cite{FMRT-cmp} complete somehow the picture of the non-equilibrium dynamics of
the East model up to the equilibrium scale $L_c=1/q$.  The question of a mathematically rigorous description
of the stationary dynamics, for which the typical domain has length
proportional to $L_c=1/q$, remains open. In
their seminal paper \cite{Aldous} Aldous and Diaconis proposed the
following very appealing conjecture.
As $q\downarrow 0$ the vacancies of the stationary East process in $[0,+\infty)$, after
rescaling space by $q$ and speeding up the process by the relaxation time,
converge to a limit point process $X_t$ on $[0,+\infty)$ which can be
described as follows:\\
(i) At fixed time $t$, $X_t$ is a Poisson point process of rate $1$.\\
(ii) For each $\ell>0$, with a positive rate depending on
  $\ell$ each particle deletes all particles to its
  right up to a distance $\ell$ and replaces them by new sample of the
  Poisson process.

Proving the conjecture is certainly one of the challenging open
problems for the model.  

\subsection{Model and notation}
\noindent 
We will consider a reversible interacting particle system on finite
intervals $\L$ of $\bbN:=\{1,2,\dots\}$ of the form
$\L=\{a,a+1,\dots,a+L-1\}$ (for shortness
$\L=[a,a+L-1]$ if clear from the context)
with  Glauber type dynamics
on the configuration space $\Omega_\L:=\{0,1\}^\L$, reversible with respect to  the product
probability measure $\pi_\L:=\prod_{x\in \L}\pi_x$,  where $\pi_x$ is the
Bernoulli$(1-q)$ measure. Since we are interested in the small $q$
regime throughout the paper we will assume $q < 1/2$.
\begin{remark}\label{temp}
Sometimes in the physical literature the parameter
$q$ is written as $q=\frac{e^{-\beta}}{1+e^{-\beta}}$, $\beta$ being
proportional to the inverse temperature,
so that the limit $q\downarrow 0$ corresponds to the zero temperature limit.
\end{remark}
Elements of $\O_\L$ will usually be denoted by the Greek letters
$\sigma,\eta,\dots$ and $\s_x$ will denote the occupancy variable at
the site $x$. The configuration after flipping the spin on site $x$ will
be denoted by $\s^x$,
\begin{align}
  \label{eq:flip}
  \s^x_y = 
  \begin{cases}
    \s_y & \textrm{if } y \neq x\,, \\
    1-\s_y & \textrm{if } y = x\,. 
  \end{cases}
\end{align}
The restriction of a configuration $\s$ to a subset $V$
of $\L$ will be denoted by $\s_V$.  In the sequel it will be useful to
use the convention that $\s_{a-1}\equiv 0$, \ie there is a fixed
vacancy on the left of the interval. 

The East process (a continuous time Markov chain) can be informally described as follows. Each vertex
$x\neq a$ waits an independent mean one exponential time and then, provided
that the current configuration $\sigma$ satisfies the constraint
$\sigma_{x-1}=0$, the value of $\sigma_x$ at $x$ is refreshed  and set
equal to $1$ with probability $1-q$ and to $0$ with probability
$q$.  The leftmost vertex $x=a$ is unconstrained.
Two configurations $\s,\s'$ are said to be neighbors under the East dynamics if
  there is a non-zero probability rate $K(\s,\s')$ of making a transition directly between them.
  Therefore two configurations $\s,\s'$ are neighbors if they differ only in a
  single coordinate $x$ and for this coordinate $\s_{x-1}=\s_{x-1}'=0$ .
\begin{remark}
Sometimes in the literature one refers to the East process as the
above process but with the constraint at $x$ 
satisfied iff the vertex immediately to \emph{right} of $x$ is empty. Of
course the two processes are equivalent under the mapping $x\mapsto
-x$. We refer to \cite{East-Rassegna} and \cite{Ritort}
for mathematical and physical background. 
\end{remark}
The associated infinitesimal Markov generator $\cL_\L$
is given by
\begin{align}
\label{thegenerator}
\cL_\L f(\sigma)
=
\sum_{x\in \L}c^\L_{x}(\sigma)\bigl[\pi_x(f)-f\bigr](\sigma) 
\end{align}
where 
\begin{equation}
\label{fin-constr}
c^\L_x(\sigma):=
\begin{cases}
1-\sigma_{x-1} & \text{if $x\neq a$}\\
1 & \text{if $x=a$} 
\end{cases}
\end{equation}
 encodes the constraint and
$\pi_x(f)$ denotes the conditional mean $\pi_\L(f\tc \{\s_y\}_{y\neq
  x})$. 

The quadratic form or \emph{Dirichlet} form associated to $-\cL_\L$
will be denoted by $\cD_\L$ and takes the form
\begin{align}
  \label{eq:DirForm}
  \cD_\L(f)=  \pi_\L\left( f \bigl( -\cL_\L f\bigr)\right)= \sum_{x\in \L}\pi_\L\left(c^\L_x \var_x(f)\right)
\end{align}
where $\var_x(f)$ denotes the conditional variance
$\pi_x(f^2)-\pi_x(f)^2$ given $\{\s_y\}_{y\neq x}$.

When the initial distribution at time $t=0$ is $\nu$ the law and
expectation of the
process will be denoted by
$\bbP^\L_\nu$ and $\bbE^\L_\nu$ respectively. If $\nu=\d_{\s}$ we write
$\bbP^\L_\s, \bbE^\L_\s$.  

In the sequel it will  be quite useful to isolate some special
configurations in $\O_\L$.
We will denote by $\mathds{1}{0} $ the configuration with a single
vacancy located at the right end of $\L$ and by $\mathds{1}$ the configuration
  with no vacancies.
Also we define $Z_n(\L) \subset \O_\L$ to be the set of configurations in
  $\O_\L$ with at most $n$ vacancies. 

\subsection{Graphical construction and basic coupling}
\label{graphical}
Here we recall a standard graphical construction which 
defines on the same probability space the finite volume East process for \emph{all} initial
conditions.
Using a standard percolation argument \cites{Durrett,Liggett2} together with the fact
that the constraints $c_x^\L$ are uniformly bounded and of finite range,
it is not difficult to see that the graphical construction can be extended without problems also to
the infinite volume case.

To each $x\in \L$ we associate a mean one Poisson process and, independently, a family of
independent Bernoulli$(1-q)$  random variables $\{s_{x,k}:k\in
\bbN\}$. The occurrences of the Poisson process associated to $x$ will
be denoted by $\{t_{x,k}: k\in \bbN\}$. We assume independence as $x$
varies in $\L$. Notice that with probability one all the
occurrences $\{t_{x,k}\}_{k\in \bbN,\, x\in \L}$ are
different. This defines the probability space. The corresponding
probability measure will be denoted by $\bbP^\L$. 

Given $\h\in \O_\L$ we construct a continuous time Markov chain
$\{\eta(s)\}_{s\ge 0}$
on the above probability space, starting at $t=0$ from $\eta$, according to the
following rules.
At each time $t=t_{x,n}$ the site $x$ queries the state of its own
constraint $c_x^\L$. If and only if the constraint is satisfied then $t_{x,n}$ is called a \emph{legal ring} and
at time $t$ the configuration resets its value at site $x$  to the value of the corresponding Bernoulli variable
$s_{x,n}$. It is easy to check that the above construction actually gives a
continuous time Markov chain with generator \eqref{thegenerator}. We will refer in the sequel to the above construction as
the \emph{basic coupling} for the process. 

\begin{remark}
Notice that
the rings and coin tosses at $x$ for $s\le t$ have no influence whatsoever on the
evolution of the configuration at the site $x-1$ which determines the
constraint $c^\L_x$ and thus they have no influence of whether a ring at $x$ for $s > t$ is legal
or not.   
\end{remark}
A first immediate consequence of the construction is the following
characterization of the \emph{coupling} time of the chain.
Starting from $\xi\equiv \mathds 1$
define $\t^{(x)}$ as the first legal ring in $x\in \L$.
Then elementary induction gives that, for any $x\in \L$, any $\eta$ and any $t\ge
\t^{(x)}$,
\begin{equation} 
\label{basic-coupling}
\eta_y(t)=\xi_y(t)\quad \forall y \leq x . 
  \end{equation}
In particular 
\[
\bbP^\L\left(\exists \eta,\eta': \ \eta(t)\neq \eta'(t)\right)\le
\bbP^\L_{\mathds 1}(\t^{(L)}> t).
\]
The second consequence of the basic coupling is 
the following property (see \cite {FMRT-cmp}*{Lemma 2.2}).
Fix  $1< b<c\leq L$ 
in $\bbN$ and let 
\[
\L=\{1,2,\dots, L\},\quad \L_1=\{1,2,\dots,b\},\quad
\L_2=\{b+1,b+2,\dots ,c\}.
\]  
For any $\eta\in\Omega_\L$ take two events ${\mathcal{A}}$
and ${\mathcal{B}}$, belonging  respectively to the $\s$--algebras
generated by $\{\eta_x(s) \}_{s\leq t,\, x\in \L_1}$ and
$\{\eta_x(s)\}_{s\leq t,\, x\in \L_2}$. Then,
\begin{align*}
\bbP_{\eta}^{\L}({\mathcal{\cA}})&=\bbP_{\eta_{\L_1} }^{\L_1}({\mathcal{A}})\\
  \bbP_{\eta  }^{\L}({\mathcal{A}}\cap{\mathcal{B}}\cap\{
\eta_{b}(s)=0 ~\forall s\leq t\})
 &=\bbP_{\eta_{\L_1} }^{\L_1}({\mathcal{A}}\cap\{\eta_{b}(s)=0 ~\forall s\leq t\})\ \bbP_{\eta_{\L_2} }^{\L_2}({\mathcal{B}})\,.
\end{align*}
The last, simple but quite important consequence of the graphical
construction is the following one. Assume
that the zeros of the starting configuration $\s$ are labeled in
increasing order as $x_0,x_1,\dots,x_n$ and define $\t$ as the first time
at which one of the $x_i$'s is killed, \ie the occupation variable there flips to
one. Then, up to time $\t$ the East dynamics factorizes over the East
process in each interval $[x_i, x_{i+1})$.

\subsection{Ergodicity and some background}

The finite volume East process is trivially exponentially ergodic because the
variable $\eta_{a}$ at the beginning of the interval $\L$ is
unconstrained ($c_a^\L(\s)\equiv 1 $). The infinite
volume process in $\bbZ$, which can be constructed by standard
methods \cite{Liggett1}, is also ergodic in the sense that $0$
is a simple eigenvalue of the corresponding generator $\cL$ thought of as a
self adjoint operator on $L^2(\O,\pi)$ \cite{CMRT}. As far as more
quantitative results are concerned we recall the following (see  \cite{CMRT} for part (i) and \cite{CMST} for part (ii)).
\begin{proposition}\label{prem}
  \begin{enumerate}[(i)]
  \item The generator $\cL$ has a
    positive spectral gap $\l=\l(q)$. Moreover 
    \begin{equation*}
      \lim_{q\downarrow 0}\log(\l^{-1})/\left(\log(1/q)\right)^2=(2\log 2)^{-1}.
    \end{equation*}
and for any interval $\L$ the spectral gap of the finite volume
generator \eqref{thegenerator} is not smaller than $\l$.
\item Assume that the initial distribution $\nu$ is a product
  Bernoulli($\a$) measure, $\a\in (0,1)$. Then there exists $m\in (0,\l]$ and for any function
  $f$ depending on finitely many variables there exists a constant $C_f$ such that
  \begin{equation*}
    |\bbE_\nu[f(\s(t))]-\pi(f)|\le C_f e^{-mt}
  \end{equation*}
 \end{enumerate}
\end{proposition}
The above results show that relaxation to equilibrium is indeed taking
place at an exponential rate on a time scale $T_{\rm rel}=\l^{-1}$
which is very large and of the order of
$ e^{c \log(1/q)^2}$, $c=(2\log 2)^{-1}$, for small values of $q$.

\section{Main results}
In order to state our main findings we first need to define some appropriate
characteristic time scales associated to the East process on the interval
$\L=[a,a+L-1]$. Without loss of generality we take $a=1$. As is
apparent from their definition, they are all non-decreasing in $L$ (see Lemma
\ref{monotone}).

The first one will be the relaxation time.
\begin{definition}[Relaxation time]
The spectral gap,  $\gap (\cL_\L)$,  of the
infinitesimal generator is the smallest positive eigenvalue of
$-\cL_\L$ and it is given by the variational principle
\begin{align}
\label{eq:gap}
\gap ( \cL_\L):= \inf _{ \substack{f :\O_\L \to \bbR\\ f \text{ non constant} } } \frac{ \cD_\L(f) }{\var_\L(f) }\,.
\end{align}
The relaxation time $T_{\rm rel} (L)$ is defined as the inverse of the spectral gap:
\begin{equation}\label{rilasso}
T_{\rm rel} (L)= \frac{1}{ \gap (\cL_\L)}\,.
\end{equation}
\end{definition}
Our second time scale is the \emph{mixing time} of the process in $\L$.
\begin{definition}[Mixing time]
Writing $\| \cdot \|_{\rm TV}$ for  the total variation distance, the mixing time $T_{\rm mix}(L)$ is defined as
\begin{equation}\label{mescolo}
 T_{\rm mix}(L)= \inf \left \{t\ge 0: \max_\eta\|  \bbP^\L_\eta(\eta_t=\cdot) -\p_\L(\cdot)\|_{TV}\leq 1/4 \,\right\}\,.
 \end{equation} 
\end{definition}
It is well known (see e.g. \cite{Saloff}) that 
\begin{equation*}
T_{\rm rel}(L)\le T_{\rm mix}(L)\le T_{\rm
  rel}(L)\left(1+\frac{1}{2} \log(\frac{1}{\pi^*})\right)
\end{equation*}
where $\pi^*:=\min_\s\pi_\L(\s)$.
The last important characteristic time is an expected hitting time.
\begin{definition}[Mean hitting time]
Let $\t_{\h_L=1}$ be the hitting time of the set
$\{\eta: \eta_L=1\}$.
Then
\begin{equation}\label{raggiungo}
T_{\rm hit}(L):=\bbE^\L_{\mathds{1} 0} \bigl[ \t_{\eta_L=1}\bigr ]\,.
\end{equation}
\end{definition}
To understand the relevance of the last time scale, let us suppose to start
the process 
from a generic configuration $\eta$ such that, for some $x<y\in \L$,
$\eta_x=\eta_y=0$ while $\eta_z=1$ for all $z\in (x,y)$. 
The configuration at $x+1$ is unconstrained until the first time the vacancy
at $x$ disappears. In particular, the vacancy at $x$ can
create waves of vacancies to its right which could remove the vacancy
at $y$. 
Conditioned on the vacancy at $x$ surviving, the
expected time to remove (or kill) the vacancy at $y$ is given by
$T_{\rm hit}(\ell)$ with $\ell=y-x$. Thus the time scales $T_{\rm
  hit}(\cdot)$ can be used as a first attempt to measure the
\emph{lifetime} of the \emph{domains}.
\begin{remark}
In the sequel we will be interested in the above time scales as
functions of the facilitating density $q$ as $q\downarrow 0$. The
dependence on $q$ will, in general, be twofold: that due to the East
dynamics and that due to a (possible) dependence of the
length scale $L$ on $q$.  
\end{remark}

\subsection{Bounds on the characteristic times}

Roughly speaking our first result
states that, for $q\ll 1$ and for all \emph{length scales} $L\le {\rm const}\times 1/q$, the
above \emph{time scales} are all of the same order as a function of
$q$. 
\begin{definition}
\label{def:equiv}
Given
two positive functions $f,h$ on the interval $(0,1)$ we will write $f\asymp h$
if 
\[
0<\liminf_{q\downarrow 0} \frac{f(q)}{h(q)} \le \limsup_{q\downarrow 0}\frac{f(q)}{h(q)} <+\infty. 
\]  
If instead there exists a positive constant $\b$ such that 
\[
\liminf_{q\downarrow 0}q^\b \frac{f(q)}{h(q)}>0
\]
then we will write $f\succ h$.
\end{definition}
\begin{remark}
Notice that the relation $\succ$ requires a rather strong divergence
of the ratio $f(q)/h(q)$ as $q\downarrow 0$. Such a choice was
motivated by some work on the topics discussed here which  appeared in
the physical literature (see \cites{SE1,SE2}).  
\end{remark}
\begin{theoremA}[Equivalence  of  characteristic times up to scale
  $O(1/q)$\footnote{We recall that $f=O(g)$ means that there exists a
    constant $C>0$ such that $|f|\le C\,g$. }]\label{equo} For each $L$ 
\begin{equation}\label{colazione}
(1-q)^L \ T_{\rm hit}(L) \leq T_{\rm rel}(L) \leq T_{\rm mix}(L)\leq 4
T_{\rm hit}(L)
\end{equation} 
In particular, if $L=O(1/q)$,  then  $T_{\rm rel} (L)
\asymp T_{\rm mix}(L)\asymp T_{\rm hit}(L) $.
\end{theoremA}
\begin{remark}
The equivalence between $T_{\rm mix} (L)$ and $T_{\rm hit}(L)$ agrees
with a recent  general  result \cites{O,PS} roughly  saying that  the
mixing time of a Markov chain  coincides with the mean hitting time of
some likely (w.r.t. $\p_\L\,$) set. In our case, the likely set is simply the event $\{\eta_L=1\}$.
Further information on $T_{\rm hit}(L)$  is given in Section \ref{hit},
where also its equivalence with another hitting time is established.
\end{remark}
Having established the above equivalence, the next important
question concerns the dependence of the time scales on the 
length scale $L$. In particular it is of interest to know when $T_{\rm rel}(L)\succ T_{\rm
  rel}(L')$, for $L >L'$ up to the equilibrium scale $1/q$.  
It turns out that this problem is
quite non-trivial because of a rather subtle interplay between the
contribution to the relaxation time coming from energy barriers and the
contribution due to the 
entropy ({\it i.e} the number of ways of
overcoming the energy barrier). 

It is useful to first recall some known previous bounds on
$T_{\rm rel}(L)$. In the sequel, for any $L\ge 1$, we will define $n =n(L):= \lceil
\log_2 L \rceil $ where $\lceil x \rceil$ denotes the ceiling of $x$, 
\ie the smallest integer number equal to or larger than $x$. In particular 
$2^{n-1} <L \leq 2^n$. If clear from the context the dependence on $L$
of the integer $n$ will be omitted.

Using rather simple energetic considerations together with a key
combinatorial result for the East model \cite{CDG}, it  holds that, for all small enough $q$,
\begin{equation}\label{sciroppo}
\frac{ c(n)
}{q^n}  \leq T_{\rm rel} (L) \leq \frac{c'(n)}{q^n}
\end{equation}
for suitable  positive constants $c(n), c'(n)$ depending \emph{only} on
$n$ and satisfying 
\[
\lim_{n\to \infty}c(n)=\lim_{n\to \infty}c'(n)=+\infty.
\] 
The upper bound was proved in \cite{FMRT-cmp} while the lower bound
follows from Lemma \ref{imbottigliato}. The above bounds turn out to be quite precise for $n$ (i.e. $L$) fixed and
$q\downarrow 0$. If instead $n=n(q)$ depends on $q$ and it diverges as
$q\downarrow 0$, then the above
estimates deteriorate quite a bit and a more refined analysis is required. 

It was shown in \cite{Aldous} firstly and then in
\cite{CMRT} that $\sup_L T_{\rm rel}(L)<\infty$ for all $q\in (0,1)$. In particular, in
\cites{CMRT,CMST} it was proved that, for each $\delta>0$, 
there exists a positive constant $C$ such that, for all small
enough $q$,
\begin{equation}\label{stima_alto}
C^{-1}q^{2}q^{-n_*/2}\leq
 T_{\rm rel}(L)\leq  Cq^{-n_*/(2-\delta) } \,, \qquad n_* = \log_2(1/q) \,.
\end{equation}
\begin{remark}
Of course the same result does not apply to the mixing time $T_{\rm
  mix}(L)$ or to the
mean hitting time $T_{\rm hit}(L)$. Using the fact that the jump rates of the process
are uniformly bounded together with , it is quite simple to show that the latter time
scales grow at least \emph{linearly} with $L$ when $L\gg 1/q$. On the other hand, using $\sup_L T_{\rm rel}(L)<\infty$ one immediately concludes that $\sup_L L^{-1}T_{\rm mix}(L) < \infty$ and similarly for $T_{\rm hit}(L)$. That is of
course not in contradiction with the equivalence with $T_{\rm
  rel}(L)$ for scale $L$ up to equilibrium scale $1/q$.   
\end{remark}
Unfortunately none of the above estimates is able to settle the question of 
whether $T_{\rm rel}(L)\succ T_{\rm rel}(L')$ or not in a satisfactory way and much
more refined bounds are required. Our second result is a step forward
in this direction.

\begin{theoremA}[Bounds on the relaxation time]\label{paletti}
Given $d>0$ there exist constants $\a,\a'$ depending only on $d$ such that, for all $L\in [1,d/q]$,
\begin{equation}\label{piscina}
  \frac{n!}{q^n 2^{\binom{n}{2} }}q^{\a} \leq T_{\rm rel } (L) \leq
  \frac{n!}{q^n 2^{\binom{n}{2} }}  q^{-\a '},\quad n=\lceil \log_2 L\rceil .  
\end{equation} 
\end{theoremA}
\subsubsection{Some heuristics behind
\eqref{piscina}}\label{heuristic}
Since the equilibrium vacancy
density $q$ is very small, most of the non-equilibrium
evolution will try to remove the excess of vacancies present in the initial
distribution and will thus be dominated by the coalescence of domains (intervals separating two consecutive vacancies).
Of course this process must necessarily occur in a kind
of cooperative way because, in order to
remove a vacancy, other vacancies must be created nearby (to its
left). Since the creation of vacancies requires the overcoming
of an \emph{energy barrier}, in a first
approximation the
non-equilibrium dynamics  of the East model for $q\ll 1$ is driven by
a non-trivial energy landscape.

In order to better explain the structure of this landscape suppose
that we start from the configuration $\mathds{1} 0$
with only a vacancy at the right end of the interval $\L$.
In this case a nice combinatorial argument (see
\cite{CDG} and also \cite{SE2,SE1}) shows that, in order to remove the
vacancy within time $t$, there must exists $s\le t$ such
that the number of vacancies inside $\L\setminus \{L\}$ at time
$s$ is at least $n$. A simple comparison with the stationary East
model, using the fact that $\pi(\mathds 1)=1-o(1)$, shows that at any given time
$s$ the probability of observing $n$ vacancies in $\L\setminus \{L\}$ is
$O(q^n)$. Therefore, in order to have a non negligible probability of
observing the disappearance of the vacancy at $L$, one expects to have
to wait an
\emph{activation time} $t_n=O(1/q^n)$. In a more physical language the energy barrier which the system
must overcome has height $n$. 

The above heuristics in a sense explains the first main contribution
$(1/q)^{n}$ appearing in \eqref{piscina}. The other main contribution,
$n!/2^{\binom{n}{2} }$, is much more subtle and more difficult to
justify heuristically. It is an
\emph{entropic} term related, in some sense, to the cardinality of the set   $V(n)$ of configurations on $\L\setminus \{L\} $
which can be reached from the configuration $\mathds{1}$ using at most $n$ vacancies in $\L\setminus \{L\}$. Equivalently, $V(n)$ is given by the  configurations in $\L\setminus \{L\}$ that can be reached from $\mathds{1}$ through a path in $Z_n(\L\setminus \{L\})$. 
 The cardinality  $|V(n)|$ satisfies the inequalities (see \cite{CDG})
\[
c_1^n n! 2^{\binom{n}{2} }\le  |V(n)| \le c_2^n n! 2^{\binom{n}{2} }
\]
where $c_1,c_2$ are positive constants in $(0,1)$. 

A first naive guess would be that the actual relaxation
time is the activation time $t_n$ reduced by a factor proportional to
$|V(n)|^{-1}$ (see \cite{CMST} for a rigorous lower bound on $T_{\rm
  rel}(L)$ based on this idea).  Notice however that the true
reduction factor in \eqref{piscina} is much smaller and equal to
$2^{\binom{n}{2} }/n!$. Thus only a tiny fraction of the
configurations reachable with $n$ vacancies actually belong to the
energy barrier.  Many configurations with $n$ vacancies will return
quickly to $\10$ before removing the vacancy at $L$, and are therefore not
typically visited during an excursion which overcomes the energy
barrier.


\subsection{Time scale
separation and dynamic heterogeneity}

Theorems \ref{equo} and \ref{paletti} have some interesting
consequences on two basic and strongly interlaced questions concerning the non-equilibrium
dynamics of the East model at low $q$. The first one is whether the characteristic time scales
corresponding to different length scales have the same scaling, as 
$q\downarrow 0$, or not. As we will see numerical simulations in this case can
be quite misleading. The second question is whether and to what extent
we should expect \emph{dynamic heterogeneity} in the model. To
simplify the notation we do not write explicitly the integer part
when the meaning is clear. 

\begin{theoremA}[time scale separation up to mesoscopic scales] 
\label{bubu} 
The following holds:
\begin{enumerate}[(i)]
\item
Given $L',L$ independent of $q$, $T_{\rm rel} \left( L' \right) \succ
T_{\rm rel} \left( L\right)$
if and only if $ \lceil \log_2 L' \rceil >   \lceil \log_2 L \rceil$.

\item
Given $\gamma\in (0,1)$, there exists  $\l=\l(\g)>1$ such that, for
all  $L=d/q^\g$, $d>0$, 
\begin{equation}
\label{fabio1}
T_{\rm rel} \left( \l L\right) \succ  T_{\rm rel} \left( L\right)\,.
\end{equation}
Moreover $\l=2$ when  $ \gamma < 1/2$. 
\end{enumerate}
\end{theoremA}
While at finite lengths (\emph{i.e independent of
  $q$})  the question of time scale separation is completely
characterized (see (i) above), for mesoscopic lengths of order $O(1/q^\g)$,
$\gamma \in (0,1)$, our knowledge is less detailed. Although
time scale separation occurs between length scales whose ratio
is above a certain threshold $\l$ (see (ii) above), we would like to know,
for example, 
if it occurs in a ``continuous'' fashion. By that we mean the following. 
\begin{definition}[Continuous time scale separation]
Given $\g\in (0,1]$ we say that \emph{continuous time scale separation}
occurs at length scale $1/q^\g$ if  $T_{\rm rel} \left( d'/ q^\gamma  \right) \succ  T_{\rm rel} \left( d/q^\gamma \right)$ whenever $d'>d$.
\end{definition}
In \cites{SE1,SE2} a continuous time scale separation was conjectured
for $\g=1$, {\it i.e.} for length scales of the order of the
equilibrium inter-vacancy distance. The following hypothesis was
put forward in \cites{SE1,SE2} on the basis of numerical simulations
and was the base of the so-called
\emph{super-domain dynamics} proposed by Evans
and Sollich to describe the time evolution of the stationary East model.
\smallskip

{\bf Time scaling hypothesis}. \emph{Continuous time scale separation occurs at length scale
$1/q$. Moreover there exists a
strictly increasing positive function $f:(0,+\infty)\mapsto \bbR$ such that, as $q\downarrow 0$, 
 \begin{equation}\label{levissima}
 T_{\rm rel} \left(d/q\right) = \left(1/q\right)^{f(d) +o(1) } T_{\rm rel} \left(1/q \right)
 \end{equation}
}
\begin{remark}
Strictly speaking the above hypothesis was formulated with $T_{\rm
  hit}(L)$ in place of $T_{\rm rel}(L)$. However, thanks to Theorem
\ref{equo}, the two formulations are completely equivalent.  
\end{remark}
The next result shows that the above hypothesis is false.
\begin{theoremA} [Absence of  time scale separation on scale  $1/q$] \label{sorpresa}
There is no time scale separation at the equilibrium scale. More precisely
\begin{equation}
  \label{salvia1}
T_{\rm rel}(d/q)\asymp T_{\rm rel}(d'/q)\quad \forall d,d' >0.
\end{equation}
In particular
\begin{equation}\label{salvia2}
\lim _{q \downarrow 0 } \left(\frac{  T_{\rm rel} (  d'/q  ) }{ T_{\rm
      rel} ( d/q  ) }\right)^{\frac{1}{\log(1/q) }}=1\,, \qquad
\forall d'>d,
\end{equation}
so $f$ in \eqref{levissima} is identically zero.\end{theoremA}
Our last result concerns \emph{dynamic heterogeneity}. Roughly
speaking it says that the following holds for $\g\in [0,1), d>0$: 
\begin{itemize}
\item Domains
much shorter than $d/q^\g$ are very
unlikely to be present at time $T_{\rm rel}(d/q^\g)$.
\item An initial domain larger than 
$d/q^\g$ is likely to still be present at  time $T_{\rm rel}(\epsilon d/q^\g)$
for some $\epsilon$ small enough.
\end{itemize}
\begin{theoremA}
\label{hetero}
Let $\L=[1,L]$ and fix $\g\in
  [0,1), d>0$.

(i) Assume $L=o(1/q)$.
  Then, as $q\downarrow 0$ and with $t=T_{\rm rel}(d/q^\g)$,
\label{th:dh}  
\[
\sup_\eta \bbP_\eta^\L\Bigl(\exists z_1<z_2:\ z_2-z_1\le \epsilon d/q^\g
\ \text{and } \eta_{z_1}(t)=\eta_{z_2}(t)=0\Bigr) =o(1) 
\]
for all $\epsilon$ small enough.

(ii) Let $d/q^\g\le L \le 1/q$.
Choose 
an initial configuration $\eta$ such that $\eta_L=0$ and $\eta_x=1$
for all $L-d/q^\g \le x <L$. Then, as
$q\downarrow 0$ and with
$t=T_{\rm rel}(\epsilon d/q^\g)$,  
\[
\bbP_\eta^\L\Bigl(\eta_L(t)=0\Bigr)=1-o(1)
\]    
for all $\epsilon$ small enough.
\end{theoremA}

\section{Preliminary results on hitting times}
\label{hit}

With $\L=[1,L]$ we prove some preliminary results about  the hitting time $\t_{\eta_L=1}$ under $\bbP^\L_{\mathds{1}0}$ as well
of the hitting time  $\t_{\eta_L=0}$ under
$\bbP^\L_{\mathds{1}}$. The last one was also analyzed in
\cites{SE1,SE2} by simulations and assumed to be equivalent to $\t_{\eta_L=1}$. 
We also show that all the relevant time scales are
non-decreasing as a function of $L$.
To simplify notation  we write $\t_L := \t_{\eta_L =1}$
and $\hat\t_L := \t_{\eta_L =0}$.

\begin{lemma}
\label{monotone}
The time scales $T_{\rm rel}(L),
  T_{\rm mix}(L), T_{\rm hit}(L)$ are non-decreasing in $L$.  
\end{lemma}
\begin{proof}
Fix an integer $L$ and consider the East model in $\L=[1,L]$. 
Clearly the restriction of the process to the first $L-1$ sites coincides with 
the East model in $[1,L-1]$. Hence $T_{\rm
  rel}(L)\geq T_{\rm rel}(L-1)$ and similarly for $T_{\rm
  mix}(L)$. 
As far as the mean hitting time is concerned we just observe that, in order to remove the last vacancy at $L$ starting
from $\mathds{1}{0}$, we need to wait at least the first time $\hat\tau_{L-1}$ for the process to remove the initial particle at $L-1$. 
Thus 
\[
T_{\rm hit}(L)=\bbE_{\mathds{1}{0}}^\L[\tau_{L}]\ge
\bbE_{\mathds{1}}^\L[ \hat\tau_{L-1}]\,.
\] 
Let $\L' = [1,L-1]$, then as shown in Proposition \ref{dominare} below  $\bbE_{\mathds{1}}^\L[\hat\tau_{L-1}]\ge \bbE_{\mathds{1}{0}}^{\L'}[\tau_{L-1}]=T_{\rm hit}(L-1)$.
\end{proof}

\begin{proposition}\label{dominare}
The hitting time $\t_{L}$ under $\bbP^{\L}_{\mathds{1}0}$ is stochastically dominated by the hitting time
$\hat\t_{L}$ under $\bbP^{\L}_{\mathds{1}}$ and stochastically dominates  $\hat\t_{L-1}$ under $\bbP^{\L}_{\mathds{1}}$.
Moreover
\begin{equation}
\label{mucca}
\bbE ^{\L}_{\mathds{1}} \left[ \hat\t_{L-1} \right] \leq 
\bbE ^\L_{\mathds{1}0} \left[ \t_{L} \right]  \leq 5\,\bbE^{\L}_{\mathds{1}} \left[ \hat\t_{L-1} \right] \,.
\end{equation}
\begin{align}
& \bbP^\L_{\mathds{1}0 }  \left( \t_{L} >t \right) \leq (1/4) ^{\lfloor   t/T(L) \rfloor } \,, \label{tortino1}\\
& \bbP^\L_{\mathds{1}0 }  \left( \t_{L} <t \right) \leq  e
  t/T_{\rm hit}(L),\label{tortino2}
\end{align}  
where $T(L)$ is characterized by the identity 
$ P^\L_{\mathds{1}0 }  \left( \t_{L} >T(L)\right)=1/4$.
In addition,  $T(L)$ is  bounded below by $T_{\rm mix}(L)$ and  satisfies 
\begin{equation}\label{contucci}
 (1/4)T_{\rm hit}(L)  \leq T(L) \leq 4 T_{\rm hit}(L) \,,
\end{equation}
\end{proposition}
\begin{remark}
Results similar to \eqref{tortino1}, \eqref{tortino2} and
\eqref{contucci} hold for the hitting time $\hat \t_L$ under
$\bbP_{\1}^\L$.   
\end{remark}


\begin{proof}
  We first prove the stochastic domination. 
  The fact that $\hat\t_{L-1}$ is stochastically dominated by $\t_L$ is trivial, since in order 
  to create a particle at $L$ starting from $\10$, one first has to create a 
  zero at $L-1$. In particular $\bbE ^{\L}_{\mathds{1}} \left[ \hat\t_{L-1} \right] \leq 
\bbE ^\L_{\mathds{1}0} \left[ \t_{L} \right]$.
  We now prove that $\t_L$ starting from $\xi := \10$ is
  stochastically dominated by $\hat\t_L$ starting from $\hat\xi:=\1$.
  We couple the two evolutions $\xi(t)$ and $\hat\xi(t)$, starting
  from $\mathds{1}0$ and $\mathds{1}$ respectively,  as follows:
  \begin{itemize}
  \item For all sites $x \in [1,L-1]$ use the basic coupling described in Section
    \ref{graphical}, so that $\xi_x(t) = \hat\xi_{x}(t)$
  \item  If at time $t$ there is a legal ring at site $L$ then, setting $p=1-q$,
    $$
    \begin{cases}
      \xi_L(t)=0\,,\qquad  \hat\xi_L(t)=1 &
      \text{ with probability $q$,}\\
      \xi_L(t)=1\,,\qquad  \hat\xi_L(t)=1 &
      \text{ with probability $p-q$,}\\
      \xi_L(t)=1\,,\qquad  \hat\xi_L(t)=0 &
      \text{ with probability $q$.}
    \end{cases}
    $$
  \end{itemize}
  We now observe that for  both processes the legal rings at site
  $L$ coincide and both processes restricted to $[1,L-1]$ are identical.
  Suppose $t$ is the first time that $\hat\xi_L(t)=0$, i.e. $t =
 \hat\t_L$. 
 Then the third case in  the above list happens, in particular  $\xi_L(t)=1$. 
 This implies that $t \geq \t_L$. This concludes our proof of the stochastic domination.

We now prove the second upper bound in \eqref{mucca}. Equivalently,  we need   to prove that  $\bbE ^\L_{\mathds{1}0} \left[ \t_{L} \right]  \leq 5\,\bbE^{\L}_{\mathds{1}0} \left[ \hat\t_{L-1} \right]$, hence we consider the East process on $\L$ starting from $\mathds{1}0$.
Let $\t$ be the waiting time after $\hat\t_{L-1}$ for the first ring on site $L-1$ or $L$ (note that this is not necessarily legal if the first ring occurs at $L-1$), then $\t \sim \text{Exp}(2)$ and is independent of  $\hat\t_{L-1}$.
Let $A$ be the event that this ring occurs on site $L$, not $L-1$, and
it updates to a $1$, so that $A := \{ \eta_{L}(\hat\t_{L-1}+\tau) = 1 \}$.
Then $\chi := \1_A\sim \text{Binomial}(p/2)$ and is independent of $\t$ and $\hat\t_{L-1}$ (see the graphical construction in Section \ref{graphical}).
If $\chi = 0$ then we couple the process started from $\eta(\hat\t_{L-1}+\tau)$ with that started from $\10$ according to the basic coupling.
Then by \eqref{basic-coupling} the hitting time of $\{\eta:\ \eta_L = 1\}$
starting from $\eta(\hat\t_{L-1}+\tau)$ is dominated by the hitting
time of $\{\eta:\ \eta_L=1\}$ started from $\10$, call this  $\t'$. 
Clearly $\t'$ is independent of $\chi$.
It follows that
$$
\t_L \leq \hat\t_{L-1} + \chi\, \t + (1-\chi)\,\t'\ .
$$ 
Taking expectation (where $\bbE^\L$ refers to the basic coupling) we have,
\begin{align*}
  \bbE_{\10}^\L[\t_{L}] &\leq \bbE_{\10}^\L[\hat\t_{L-1}] + \bbE^\L[\chi\, \t] + \bbE^\L[(1-\chi)\,\t']\\
  &\leq \bbE_{\10}^\L[\hat\t_{L-1}] + \frac{p}{4}+ \left( 1 -\frac{p}{2} \right) \bbE_{\10}^\L[\t_{L}]
\end{align*}
so,
\begin{align*}
  \bbE_{\10}^\L[\t_L] &\leq \frac{2}{p}\bbE_{\10}^\L[\hat\t_{L-1}] + \frac{1}{2}\leq \left(  \frac{2}{p}+ \frac{1}{2} \right)\bbE_{\10}^\L[\hat\t_{L-1}] 
\end{align*}
which clearly leads to the second upper bound in \eqref{mucca} (recall that $q \leq 1/2$).
 
In order to prove \eqref{tortino1} it is enough to show that,  given $t,s\geq 0$, 
 \begin{equation}
   \label{tommy} 
   \bbP^\L_{\mathds{1}0} \bigl( \t_{L} >t+s\bigr) \leq \bbP^\L_{\mathds{1}0} \bigl( \t_{L} > t\bigr) 
   \bbP^\L_{\mathds{1}0} \bigl( \t_{L} > s\bigr) \,.
 \end{equation}
 Indeed, by the Markov property we can write 
 \begin{gather*}
   \bbP^\L_{\mathds{1}0} \bigl( \t_{L} > t+s \bigr) = \bbE^\L_{\mathds{1}0 } \left[\mathds{1}_{\{
       \t_L >t\}}\, \bbP^\L_{\eta(t)} \bigl( \t_{L} > s \bigr)
   \right]\\
   \leq \bbP^\L_{\mathds{1}0} \bigl( \t_{L} > t\bigr)
   \sup _{\eta \in \O_\L}  \bbP^\L_{\eta} \bigl( \t_{L} > s\bigr)
   \,.
 \end{gather*}
 To conclude we observe that the last supremum is realized by
 $\eta= \mathds{1}0$, which follows from the basic coupling described
 in Section \ref{graphical} (cf. \eqref{basic-coupling}).
 Note that by the same arguments \eqref{tommy} holds 
 replacing $ \bbP^\L_{\mathds{1}0}$ and $ \t_{L}$ by   $ \bbP^\L_{\mathds{1}}$ and  $\hat\t_{L}$, thus leading to the same result  \eqref{tortino1} for the other hitting time. 
 Therefore, from now on we concentrate only on $\t_{L}$ under $ \bbP^\L_{\mathds{1}0}$ (the other hitting time can be treated similarly). 
 Integrating \eqref{tortino1} one gets the left bound in \eqref{contucci}, while one derives the right bound directly by the definition of $T(L)$. 
 The fact that $T(L) \geq T_{\rm mix}(L)$ is proved in \eqref{torneo}. At this point it remains to prove 
 \eqref{tortino2} which follows from a general fact, Lemma \ref{zanzibar} below, and \eqref{tommy}.
\end{proof}

\begin{lemma}\label{zanzibar}
Let $\t$ be a positive random variable such that   $t\mapsto
\bbP(\tau\ge t)$ is continuous and sub-multiplicative:
$$ \bbP( \t >t+s) \leq\bbP( \t >t) \bbP(\t>s)\,, \qquad \forall t,s \geq 0\,.$$
Then 
\[
\bbP(\t<t) \leq  e\, t/\bbE(\tau),\quad  \forall t>0.
\]
\end{lemma}
\begin{proof}
Let $t$ be such that $\bbP(\tau > t) <1$. The sub-multiplicative property implies that 
$$ 
\bbP( \t> s) \leq \bbP(\tau > t)^{\lfloor s/t\rfloor
}\leq \bbP(\tau > t)^{s/t-1}, \qquad \forall s \geq 0\,.
$$ 
Integrating over $s$ we get
\[
\bbE(\t) \leq t\left[ \bbP(\tau > t)\log\frac{1}{\bbP(\tau > t)}\right]^{-1}.
\] In particular
$\bbE(\tau)\le e\, t_*$ where $t_*$ is such that $\bbP(\tau \ge
t_*)=e^{-1}$. Assume now $t \le t_*$. Then $\bbP(\tau \ge t) \ge e^{-1}$
and 
\begin{align*}
\bbP(\tau \le t) &\le  \log\Bigl(\frac{1}{1-\bbP(\tau\le
    t)}\Bigr)\le e\, \bbP(\tau \ge t) \log\Bigl(\frac{1}{1-\bbP(\tau\le
    t)}\Bigr) \le e\, t/\bbE(\tau) \,.
\end{align*}
Note that in the first inequality we have used that $x \leq \log \frac{1}{1-x}$ for all $x \in (0,1)$.
If instead $t\ge t^*$ then $e\, t/\bbE(\t)\ge 1$ and there is nothing
to prove.
\end{proof}

For system sizes which are much shorter than the equilibrium length
scale we observe a loss of memory property, that leads to $\t_L$ being
approximately exponentially distributed, a phenomena which is
typically associated with metastable dynamics.
\begin{lemma}
If $L =  d/q^\g$ for some $d>0$ and $\g\in[0,1)$, then 
  \begin{align}
    \label{eq:4}
    \lim_{q\searrow 0} \bbP^\L_{\10}\left(\frac{\t_L}{\bbE_{\10}^\L[\t_L]}
        > t\right) = e^{-t}\,.
  \end{align}
\end{lemma}

\begin{proof}
  Let $f(t):=\bbP^\L_{\10}\left(\t_L/\bbE_{\10}^\L[\t_L] > t\right)$,
  we show that
  \begin{align*}
    \left| f(t+s) - f(t)f(s) \right| \leq c q^{1-\g}\,,
  \end{align*}
  for some $c > 0$ independent of $q$. The result then follows by
  standard arguments (for example see \cite{Olivieri-Vares}*{Lemma $4.34$}).
  For brevity of notation let $E:=\bbE_{\10}^\L[\t_L]$. For any
  $s,t>0$ we have
  \begin{align*}
    f(t+s) = \bbP^\L_{\10}\left(\t_L/E > t+s \mid \t_L/E > s\right)f(s)\,,
  \end{align*}
  and by the Markov property,
  \begin{align*}
    \bbP^\L_{\10}\left(\t_L/E > t+s \mid \t_L/E > s\right) &= f(t)\bbP^\L_{\10}\left( \h(s E) = \10 \mid \t_L/E > s\right) +\\
    &+ \bbP^\L_{\10}\left( \{\t_L/E > t+s  \}\cap\{\h(s E) \neq \10\} \mid \t_L/E > s\right)\,.
  \end{align*}
  It follows that 
  \begin{align*}
    \left| f(t+s) - f(t)f(s) \right| &\leq 2 \bbP^\L_{\10}\left( \{\t_L/E > s  \}\cap\{\h(s E) \neq \10\}\right)\\
    &\leq \sum_{x\in\L\setminus {\{L\}}}\bbP^\L_{\10}\left( \eta_x(sE) = 0 \right) \leq d q^{1-\gamma}\,,
  \end{align*}
  where for the last inequality we use that for $t>0$ and $\eta_x(0)=1$ we have $\bbP^\L_{\eta}\left(\eta_x(t) = 0 \right) \leq q$, which is a simple consequence of the graphical construction (see for example \cite{FMRT-cmp}*{Lemma 4.2}). 
\end{proof}


%
%
%
%

\section{Proof of Theorem \ref{equo} }
  
The fact that $T_{\rm rel} (L)
\asymp T_{\rm mix}(L)\asymp T_{\rm hit}(L) $ for $L=O(1/q)$ follows trivially from \eqref{colazione}
since $p^{d/q} \sim e^{-d}$ as $q \downarrow 0$. Thus it is enough to
prove \eqref{colazione}.

We begin by recalling  a general result about hitting times (see Proposition 21 in \cite{Aldous-Fill}). 
\begin{lemma}\label{primastima} Let $A \subset \O_\L$ and $\s \in
  \O_\L$. Let $\t_A$ denotes the hitting time of the set $A$. Then
\begin{equation}
 \bbP^\L_{\s}(\t_A>t)\leq \frac{\p(A^c)}{\p(\s)}
e^{-t\, \p(A)/  T_{\rm rel} (L) }\,.
\end{equation}
\end{lemma}
We can now prove the  first bound $(1-q)^LT_{\rm hit}(L)\le T_{\rm
  rel}(L)$.
As a special case of Lemma \ref{primastima} we have
$$
 \bbP^\L_{\mathds{1}0 }(\t_{\eta_L=1} >t)\leq  
p^{-(L-1)} e^{-t p  /T_{\rm rel}(L)}\,.
$$
Thus
\begin{align*}
T_{\rm hit}(L)&= \bbE^\L _{\mathds{1}0} \left[ \t_{ \eta_L=1} \right] =
\int_0^\infty \bbP^\L_{\mathds{1}0 }(\t_{\eta_L=1} >t)dt \\
&\leq 
(1-q)^{-(L-1)}\int_0^\infty  e^{-t (1-q)  /T_{\rm rel}(L) }dt = (1-q)^{-L}T_{\rm rel}(L)\,.
\end{align*}
The \emph{second bound} $T_{\rm rel}(L) \le T_{\rm mix}(L)$ is a general fact for reversible Markov chains (see
e.g. \cite{Saloff}).

To prove the \emph{last bound} $T_{\rm mix}(L)\le 4 T_{\rm hit}(L)$ we
use a coupling argument. To this aim recall the basic coupling
described in Section \ref{graphical} together with the definition of
the ``legal times'' starting from $\1$, $\t^{(x)}$, $x\in \L$. A standard result (see e.g. \cite {Levin2008}*{Cor. 5.3}) gives
  \begin{equation}\label{onda}
  \Delta(t):= \max_{\eta}\|\bbP_\h^\L(\eta(t)=\cdot) - \pi(\cdot)\|_{\rm TV} \leq \max_{\sigma,\sigma'}\bbP^\L\left( \sigma(s) \not= \sigma'(s)  \; \forall s \in [0,t]\right)\,.
  \end{equation}
Using \eqref{basic-coupling} together with \eqref{onda} we get  
  \begin{align}
\label{torneo}
  1/4 = \Delta(T_{\rm mix}(L) ) 
\leq  \bbP^\L _{\mathds{1}} \bigl( \t^{(L)} >T_{\rm mix}(L)\bigr)\,.
   \end{align}
Hence
\begin{gather*}
T_{\rm hit}(L)\ge T_{\rm mix}(L)\, \bbP^\L _{\mathds{1}0} \bigl(
\t_{\h_L=1}\ge T_{\rm mix}(L)\bigr)\\
\geq \, T_{\rm mix}(L)\bbP^\L _{\mathds{1}} \bigl(
\t^{(L)} >T_{\rm mix}(L)\bigr)\ge \frac 14 T_{\rm mix}(L).  
\end{gather*}
where we used the following elementary observations; the time of the first legal ring at
$L$ does not depend on the initial value $\eta_L$ and 
starting from $\mathds 10$ the hitting time $\t_{\h_L=1}$ is not
smaller than $\t^{(L)}$.


  %
  %
  %
\section{Proof of Theorem \ref{paletti}}

\subsection{Upper bound on the relaxation time}
\label{sec:upper-bound}
The proof of the upper bound in \eqref{piscina} is based on the iterative procedure developed in \cite{CMRT}, although refinements are necessary.
We first need the following claim.
\begin{lemma}
\label{claim:1}
Given an integer $r>2$ let $\{\ell_i\}_{i=1}^r$ be defined by the inductive scheme
\begin{align}
\label{fiesole}
  \begin{cases}
    \ell_1 = 3\,,&\\
    \ell_i = 2\ell_{i-1} - \lceil \d \ell_{i-1}\rceil\,,& \textrm{ for }
    2\leq i \leq r\, ,
  \end{cases}
\end{align}
with $\d = 1/r$. 
Let also $\g_i:=T_{\rm rel}(\ell_i)$. 
Then
\begin{equation}\label{sirena1}
\g_i \le \frac{2}{1-\sqrt{\epsilon_{i-1}}} \g_{i-1}\,,\qquad  2\leq  i\leq r\,,
\end{equation}
where $\epsilon_i=(1-q)^{\lceil \d \ell_i\rceil}$.
\end{lemma}
\begin{proof} 
Let $i \geq 2$. Since $\ell_2=5>\ell_1$ one can easily prove by induction that $ \ell_{i}> \ell_{i-1}$ as follows:
$$ \ell_i - \ell_{i-1} \geq \ell_{i-1}-( \d \ell_{i-1}+1)= \ell_{i-1} (1- \d )-1> \ell_1 (1-1/3)-1>0\,.$$
Let $\L^{(i)}=[1,\ell_i]$. 
Each interval $\L^{(i)}$ can be divided into two overlapping intervals of size $\ell_{i-1}$;
\begin{align}
  \L_1 = [1,\ell_{i-1}]\,, \quad \L_2 = [\ell_i-\ell_{i-1} +1,\ell_i]\ .
\end{align}
The overlap $\D = \L_1\cap \L_2$ then contains $N_i = \lceil \d \ell_{i-1}\rceil \geq \ell_{i-1}/r$ sites.
Furthermore $\L^{(i)}$ can be divided into two disjoint intervals $\L_1$ and $\bar\L_2:= \L_2\setminus \D$.
Since the cardinalities  of $\L_1$ and $\L_2$ are both $\ell_{i-1}$, we
have 
\[
T_{\rm rel}(\L_1) = T_{\rm rel}(\L_2) =  \g_{i-1}\,.
\]
Apply now the  block  chain in which $\L_1$ goes to equilibrium with rate one
and $\bar\L_2$ does the same if and only if there is a zero in $\D \neq \emptyset$.
More precisely, the  block chain has configuration space $\O_{\L^{(i)}}$ and generator
$$ \bbL f= \bigl(\p_{\L_1} (f)-f\bigr) + c_{\bar\L_2} \bigl( \p_{\bar\L_2}(f)-f \bigr)\,, \qquad f : \O_{\L^{(i)}} \to \bbR $$
where the new constraint $c_{\bar\L^2}$ 
is defined as $c_{\bar\L^2}( \eta)= \mathds{1} \bigl\{\exists\, x \in \D\,:\, \h_x=0\bigr\}$. 
It is simple to check that the associated Dirichlet form
is given by
$$ \bbD (f)= \pi\left(\var_{\L_1}(f)+c_{\bar\L_2}\var_{\bar\L_2}(f)\right)\,.$$

 As proven in \cite{CMRT}[p. 480] the block chain has spectral gap $(1-\sqrt{\epsilon_{i-1}})$ and therefore
it satisfies the Poincar\'e inequality
\[
\var(f) \le
\frac{1}{1-\sqrt{\epsilon_{i-1}}}\,\pi\Big(\var_{\L_1}(f)+c_{\bar\L_2}\var_{\bar\L_2}(f)\Big)\,.
\]
It was proved in \cite{CMRT}*{Sect. 4} that (recall \eqref{fin-constr})
\begin{align*}
\pi\left(c_{\bar\L_2}\var_{\bar\L_2}(f)\right) &\le \g_{i-1}\sum_{x\in \L_2 }\pi\left(c_{x}^{\L_2}\var_x(f)\right)\,,\\
\pi\left(\var_{\L_1}(f)\right) &\le \g_{i-1}\sum_{x\in
  \L_1}\pi\left(c_{x}^{\L_1}\var_x(f)\right)\,.
\end{align*}
Therefore
\begin{gather*}
\var(f) \le
\frac{\g_{i-1}}{1-\sqrt{\epsilon_{i-1}}}\,\pi\left(\sum_{x\in
  \L_1 }  c_x^{\L_1}\var_x(f) + \sum_{x\in
  \L_2}c_x^{\L_2}\var_x(f)\right)\\\le \frac{2 \g_{i-1}}{1-\sqrt{\epsilon_{i-1}}}\,\cD_{\L^{(i)}}(f)
\end{gather*}
where the factor $2$ comes from the double counting of the points in
$\D$. In conclusion
\begin{gather*}
 \g_i \le \frac{2}{1-\sqrt{\epsilon_{i-1}}}\,\g_{i-1}.
\end{gather*}
This proves \eqref{sirena1}.
\end{proof}
Next we show an important property of the length scales $ \{ \ell_i\}_{1\leq i\leq r}$ appearing in Lemma \ref{claim:1}.
\begin{lemma}\label{universo}
For all $i \leq r$ it holds that $ 2^i (1-1/r)^i \leq \ell_i \leq 2^{i+1}$.
\end{lemma}

\begin{proof}
The upper bound follows immediately by induction.
For the lower bound, it is simple to derive from \eqref{fiesole}  that $\ell_i= 2^{i}+1$ if $i \leq k:= 1+ \lfloor \log_2 (r-1) \rfloor$. 
Indeed, let $\tilde k= \max\{i:\ell_{i-1} \leq r\}$. 
Then, $\forall i: 2 \leq i \leq \tilde k$ \eqref{fiesole} becomes $ \ell_i= 2 \ell_{i-1}-1$ 
and the solution of this iterative system is $\ell_i= 2^i+1$. 
In particular, it follows that $\tilde k = k$.
If $i \geq k$ then $\ell_{i-1}/r \geq 1 $, together with \eqref{fiesole} this implies  $\ell_i \geq 2(1-1/r)\ell_{i-1}$. 
The rest of the proof is straightforward.
\end{proof}


\begin{lemma}\label{sale}
Given $r \geq 1$ such that $2^r\leq  d/q$ it holds that $ \g_r \leq q^{-c(d)} \frac{r!}{ q^r 2 ^{\binom{r}{2} }}$ for some constant $c(d)$ depending only on $d$.
\end{lemma}
\begin{proof} 
Recall the definition of $\e_i$ in Lemma \ref{claim:1}. 
Using the bounds  $(1-x) \leq e^{-x} $  and $(1-e^{-x}) \geq e^{-x}
x$ for all $x > 0$ we get 
\begin{equation}\label{maremare}\frac{2}{ 1-\sqrt{\e_{i-1} } }=\frac{2}{1-(1-q)^{\lceil \ell_{i-1}/r \rceil/2} }\leq 
\frac{2}{1-e^{-q \lceil  \ell_{i-1}/r \rceil/2} }\leq
\frac{4}{q  \lceil \ell_{i-1}/r \rceil} e^{q \lceil  \ell_{i-1} /r\rceil/2} \,,
\end{equation}
for $i\leq r$.
Due to Lemma \ref{universo} we have $\ell_i \leq 2^{i+1}$. 
Hence, using the monotonicity of $\ell_j$ and that $2^r \leq d/q$, 
we get $\ell_i  \leq 2 d/q$. This allows us to conclude that 
$q \lceil  \ell_{i-1}/r \rceil  \leq q  (\ell_{i-1}/r+1 ) \leq 2d/r+q $.
This bound together with \eqref{maremare} implies that $ \frac{2}{ 1-\sqrt{\e_{i-1} } }\leq \frac{c(d) r }{q   \ell_{i-1}}$. 
Coming back to \eqref{sirena1} we conclude  that
\begin{equation}\label{ringhiera1}
 \g_r \le  \g_1 \frac{c(d)^r  r^r}{q^r \prod _{i=1}^{r-1}  \ell_{i} } \,.
 \end{equation}
Due to Lemma \ref{universo}  and since $1-x \geq e^{-2x}$ for $x$ small,  we have
\begin{equation}\label{ringhiera2}
\prod _{i=1}^{r-1}  \ell_i  \geq 2^{\binom{r}{2} } (1-1/r) ^{\binom{r}{2}}\geq 2^{\binom{r}{2} }(1-1/r)^{\frac{r^2}{2}}\geq 
2^{\binom{r}{2} }e^{-r } 
 \end{equation}
 for $r$ sufficiently large.
Since $e^{r}= 2^{r\log_2 e } \leq (d/q) ^{\log_2 e}$ and since by Stirling's formula $r^r \leq C r! \,  e^{r}$, combining  \eqref{ringhiera1} and \eqref{ringhiera2} we get the result.
\end{proof} 

We now have the necessary tools to conclude the proof of the upper bound in Theorem \ref{paletti}. 
Fix $L \leq d/q$ and set $n= \lceil \log_2 L \rceil $.
Let $c_0:= \inf\{ (1-1/k)^k \,:\, k \geq 1\}$. Note that $c_0\in
(0,1)$. We now choose $r=r_0 := n+\lceil  \log_2 (1/c_0)\rceil $, so
that $2^{r_0} c_0 \geq L$. 
By Lemma \ref{universo}, $\ell_{r_0} \geq L$.  
Since $2^{r_0}\le 4d/c_0$, by  monotonicity and Lemma \ref{sale} we conclude  
$$ T_{\rm rel} (L) \leq \g_{r_0} \leq q^{- c( 4d/c_0) }\frac{r_0!}{ q^{r_0} 2 ^{\binom{r_0}{2} }}\leq q ^{ -c'(d,c_0) }\frac{n!}{ q^{n} 2 ^{\binom{n}{2} }}\,.$$

\subsection{Lower bound on the relaxation time}
\label{lwb}
A general strategy to find a lower bound on the relaxation time is to look for a set $A$ whose boundary $\partial A$ forms a small bottleneck in the state space $\O_\L$ \cite{Levin2008,Saloff}. 
One can upper bound the spectral gap (\ie lower bound the relaxation time) by restricting the variational formula  \eqref{eq:gap} to indicator functions of sets in $\O_\L$. In this way one gets
\begin{equation}\label{biancaneve}
T_{\rm rel} (L)  \geq  \max_{A\subset \O_\L  
}\frac{ \p(A) \p(A^c) }{\cD_\L ( \mathds{1}_A) }\geq \frac{1}{2\Phi_\star} \,,
\end{equation}
where $\Phi_\star$ is the \emph{bottleneck ratio} (also known as the Cheeger or Isoperimetric constant) given by
\begin{align*}
  \Phi_\star = \min_{A\,:\,\pi(A) \leq \frac{1}{2}} \frac{\cD_\L ( \mathds{1}_A) }{\p(A)}\,,
\end{align*}
where for a given set $A$ the ratio $\cD_\L ( \mathds{1}_A) / \p(A)$ is referred to as the bottleneck ratio of the set $A$.
Due to reversibility, $\cD_\L(\mathds{1}_A)$ (called the boundary measure of the set $A$) can be written as
\begin{equation} \label{rane} 
  \cD_\L(\mathds{1}_A)=  \sum_{\eta \in A,\s\in A^c}\pi(\eta)K(\eta,\s)=
  \sum _{\eta\in \partial A} \p(\eta) K(\eta,A^c)
\end{equation} 
where
\begin{align*}
  \partial A&:= \left\{ \h \in A\,:\,\exists\ \sigma\in A^c  \text{ such that $\eta,\sigma$ are neighbors, \ie} K(\h,\s)>0 \right\}
\end{align*}
is the internal boundary of $A$  and 
\begin{align}\label{atmosfera}
  K(\eta,A^c) = \sum_{\s\in A^c}K(\eta,\s)
  =  
  \sum_{\substack{x\in [1,L]:\\c^\L_x(\eta)=1 ,\; \eta^x \not \in A}}\left\{q \eta_x +p(1-\eta_x)\right\}
\end{align}
is the total rate with which the East dynamics starting from  $\eta$ escapes from $A$.
Notice that $K(\eta,A^c) \le (\# \{\text{ of zeros in $\eta$}\}+1)\le L+1$, so that sets $A$ for which the boundary has very small equilibrium measure with respect to the full set, $\pi(\partial A)/ \pi(A)\ll 1/L$,  admit small bottleneck ratios, which are $o(1)$.

\subsubsection{A first attempt for the choice of the set $A$}
As explained in Appendix \ref{capacitino}, we expect that there exist good choices of $A$ in \eqref{biancaneve} such that the boundary $\partial A$ separates the singleton $\mathds{1}0$ from $\{ \eta_L=1\}$, \ie $\10 \in A$ and $\{ \eta_L=1\} \subset A^c$.
In  what follows we show that for certain $A$ featuring this property, $\pi(\partial A)\ll 1$ giving rise to a large relaxation time.


Our first candidate for the set $A$ is inspired by the combinatorial work in \cite{CDG} and it is defined as the set of configurations that  are
  connected to $\10$ by paths in the set $Z_{n+1}(\L)$ (the set of configurations with at most $n+1$
zeros) under the East dynamics.
In order compute the bottleneck ratio for this set we first define an auxiliary set $U_n$ which will turn out to be the internal boundary of our set $A$. 

Let $V_n$ be the set of configurations in $\{0,1\}^\bbN$
with exactly $n$ zeros that can be obtained from $\mathds{1}$ by 
the East dynamics  and by using at most n simultaneous zeros, i.e.\ 
  that can be reached from $\mathds{1} $ by a path in $Z_n(\bbN)$. 
As proven in \cite{CDG}, for all $\h \in V_n$ the zeros of $\h$  are included in $ [1,2^n-1]$.
Then, for any $L \geq 2^n$, the set $U_n$ is defined as the set of configurations $\eta$ which are obtained from the configurations in $V_n$ by flipping to zero the spin at $L$.
%
%

The following lemma follows immediately from the bounds given in \cite{CDG}:
\begin{lemma}
\label{CDGlemma} 
There exist constants $c,c'\in (0,1)$ such that, for all $L \geq 2^n$, 
\begin{equation}\label{CDGbound}
q^{n+1} p^{L-n-1} 2^{\binom{n}{2} } n!\, (c')^n \leq \p(U_n) \leq   q^{n+1} p^{L-n-1} 2^{\binom{n}{2} } n! \,c^n.
\end{equation}
\end{lemma}
The connection between the set $A$ and $U_n$ is as follows.
\begin{lemma}\label{imbottigliato}
Let $L > 2^n$ and $A\in \O_\L$ be the set of configurations that  are
  connected to $\10$ by paths in the set $Z_{n+1}(\L)$ under the East dynamics.
Then $\partial A$ is a bottleneck separating $\mathds{1}0$ from $\{ \eta_L=1\}$ and $\partial A= U_n$. 
Moreover there exists $C>1$ such that
\begin{equation}\label{stiminabis}
T_{\rm rel}(L) \geq  \frac{ \p(A) }{2\cD_\L ( \mathds{1}_A) }\geq  \frac{C^n p^n }{ q^{n+1}  2^{\binom{n}{2} } n!}\,. 
\end{equation}
\end{lemma}
\begin{proof} 
Clearly $\mathds{1}0 \in A$. The fact that  $\{\eta: \ \eta_L =1 \} \subset A^c$ follows from the above property of the vacancies of configurations in $V_n$. In fact, to have $\{\eta: \ \eta_L =1 \} \cap  A\not= \emptyset$ there would be a path $\mathds{1}0= \eta^{(1)}, \eta^{(2)}, \dots, \eta^{(k)}$ in 
 $A$  with  $\eta^{(k)}_L=1$,  with  $\eta^{(i)}_L=0$  for $1\leq i <k$ and  with $ \eta^{(k-1)}_x=0$ for $x=L-1>2^n-1$. In particular, defining $\xi_x= \eta^{(k-1)}_x$ for $1\leq x < L$ and $\xi_x =1$ for $x \geq L$, 
   one would get a configuration $\xi \in Z_n(\bbN)$ with a vacancy on the  right of $2^n-1$, which gives rise to a contradiction. 
   
We now show that $\partial A=U_n$ and $A$ admits a small bottleneck ratio.   
$U_n$ is the set of configurations in $A$ with exactly $n+1$ vacancy, 
since starting from $\10$ the East dynamics restricted to paths in $Z_{n+1}(\L)$ can not remove the vacancy at $L$.
It is trivial to check that this last set corresponds to $\partial A$. 
Indeed, if $\h \in A$ has strictly less than $n+1$ zeros, adding or removing a zero to $\h$ by a legal flip will bring us to a new configuration inside $A$ by the definition of $A$. If $\h \in A$ has $n+1$ zeros,  adding a zero to $\h$ by a legal flip will take the configuration  outside $A$, 
thus implying  $\h \in \partial A$ (note that such a transition is possible since $n<2^n < L$ and therefore $\h$ cannot be given by the empty configuration).  This completes the proof that $\partial A=U_n$.

Take $\h \in A$ with $n+1$ zeros. If  we remove a zero from $\h$ by a legal flip we remain inside $A$. 
Hence 
$$ \cD_\L ( \mathds{1}_A) =  \sum _{\h \in U_n\subset \O_\L}\p(\h)\sum _{\substack{x\in [1,L]:\\ \h_x=1, \, c^\L_x(\h)=1} }q  \leq \p (U_n)
( n+2)q \,.
$$
Since  $ \mathds{1}0 \in A$,  $\p(A) \geq \p(\mathds{1}0)=p^{L-1}q$. 
On the other hand, we know that $\{\h:\ \h_L=1\}\subset A^c$  and therefore $\p(A^c) \geq p > 1/2$.  At this point it is enough to apply \eqref{biancaneve} and \eqref{CDGbound} in order to get \eqref{stiminabis}.
\end{proof}


Taking  $n = \lceil  \log_2 1/q\rceil $  and  $L>2^n$, \eqref{stiminabis} together with the monotonicity in $L$ of the relaxation time gives a  lower bound on the infinite volume relaxation time similar to  that previously obtained in \cite{CMST}. However the bound is quite different from the actual bound claimed in Theorem \ref{paletti}. There in fact the large term $n!$ appears in the \emph{numerator} while in \eqref{stiminabis} it sits in the denominator.    
To correct this fact we need to a find a different set, call it  $A_*$, giving rise to a smaller bottleneck ratio.

\subsubsection{Definition of  $A_*$ by means of  deterministic dynamics on $\O_{\L}$}
\label{Astar}
The construction of the set $A_*$ has been inspired by capacity theory explained in Appendix \ref{capacitino}.
One wants to choose the set $A_*$ such that $\1_{A_*}$ is `close' to the minimiser in the Dirichlet principle \eqref{eq:dir-princ}, which is given by $f(\eta) = \bbP_\eta ( \tau_{\mathds{1}0 } < \tau _B)$, where $B = \{\eta_L = 1\}$  and $\tau_{\10}$ is the hitting time of $\10 \in \O_\L$.
The dynamics at small $q$, on length scales which are much smaller than the equilibrium separation of vacancies, typically act by removing vacancies.
Also (following \cite{FMRT} and \cite{SE1}) we expect vacancies with the smallest distance to their neighbouring vacancy on the left (smallest domains) to be removed first.
Following this idea, we essential say $\eta \in A_*$ if starting from $\eta$ and removing vacancies in order of their domain size we hit $\10$ before removing a vacancy at $L$.
We now proceed by giving a formal definition of $A_*$ in terms of an algorithm, which we call the \emph{deterministic dynamics} because it will approximate the order in which vacancies are typically removed by the East process for small $q$.

For each configuration $\eta\in\Omega_{\L}$ we define the  \emph{gap} at $x\in
\Lambda$, denoted  $g_x(\eta)$,  as the distance from $x$ to the nearest vacancy on the left (including the origin where the frozen zero is located):
\begin{equation} 
  g_x(\eta) :=  \min\{d>0 \sst \eta_{x-d} = 0\}\ .
\end{equation} 

Given $\eta \in \O_{\L}$, $1 \leq d \leq L$ and $x \in [1,L]$,  we define $\phi_{d,x }(\eta)\in \O_{\L} $  as the configuration obtained from $\eta$ by removing a vacancy at $x$ if one is present with gap exactly $d$, and doing nothing otherwise:
\begin{equation}
  \phi_{d,x} (\eta)_y = 
  \begin{cases}
    1 & \text{ if } y=x,\; \eta_x=0\,,\; g_x(\eta)=d\,\\
    \eta_y & \text{ otherwise}\,.
  \end{cases}
\end{equation}
The deterministic dynamics will be defined recursively by removing vacancies firstly with gap size one, starting from the right of the system and proceeding towards the origin, then removing those with gaps size two from right to left and continuing in this way.
It is therefore convenient to endow the set  $[1,L]^2$ with the following total order:
\[
(d_1,x_1) \prec (d_2,x_2) \iff d_1 < d_2 \quad \textrm{or}\quad
d_1=d_2\ \textrm{ and }\ x_1>x_2\ .
\]
Hence, we can order the elements of $[1,L]^2$ as $\a_1\prec\a_2\prec \cdots \prec \a_{L^2}$.
Notice $\a_{\cdot}$ gives rise to a bijection between $[1,L^2]$ and $[1,L]^2$ given explicitly by $\a_k = (\lceil k/L \rceil,(L-k \mod\ L) +1)$ (where we identify $\bbZ/L\bbZ$ with $\{0,1,\ldots,L-1\}$ in the usual way).
 
We are now ready to define our deterministic dynamics in discrete time $0,1,2, \dots, L^2$.  
Given the starting configuration  $\eta \in \O_{\L}$, the new configuration $\eta[k]$ at time $k=0,1,2,\dots,L^2$ is recursively defined by the rule
$$
\begin{cases}
\eta [0]= \eta\,,\\
\eta [k] = \phi _{\a_k} \left( \eta[k-1]\right) & \text { for } k =1,2, \dots , L^2\,.
\end{cases}
$$
To more easily identify the stage of the deterministic dynamics we will use the following notation:
$$ 
\dyn{\eta}{d}{x} := \eta [k] \qquad \text{ if } \qquad \a_k= (d,x)\,.
$$

\begin{remark}\label{africano}
It is simple to check that the following procedure gives an equivalent definition of the deterministic
dynamics. $\dyn{\eta}{d}{x}$ is simply the configuration obtained from $\eta$ as follows:
\begin{enumerate}
\item  Erase all the  vacancies of gap $1$ from $\eta$  (from the right to the left or simultaneously is equivalent). Call $\eta^{(1)}$ the resulting configuration.
\item In general, for $1<i\leq d-1$ erase from $\eta^{(i-1)} $ all the vacancies with gap $i$ and call the resulting configuration $\eta^{(i)}$ .
\item Erase from $\eta^{(d-1)}$ all the vacancies of gap $d$ located at $y \geq x$. The resulting configuration is $\dyn{\eta}{d}{x} $.
\end{enumerate}

\end{remark}

Having described  the  deterministic dynamics on $\O_{\L}$ we can define our set $A_*$:

\begin{definition}
\label{def:Astar}
The set $A_*\subset \O_\L$ is given by the configurations $\eta \in \O_\L$ such that 
\begin{equation}\label{pierpibello} 
  \dyn{\eta}{L-1}{1} = \mathds{1}0\, . 
\end{equation}
 Equivalently, $A_*$ is the set of all configurations $\eta$ such that
 the deterministic dynamics hits $\10$ before $\{\eta_L = 1\}$.
\end{definition}
\begin{remark}
\label{boundaryA}
 It is simple to check that  $\mathds{1}0$ belongs to $A_*$ and that
 $\{\eta_L = 1\} \subset A^c_*$.
An example of $\eta \in \partial A_*$ is shown in Fig. \ref{fig:det-dyn}.
\end{remark}
We proceed by showing that $A_*$ admits a very small bottleneck ratio.
\begin{figure}[t]
  \centering
  \mbox{\subfloat[$\eta\in A_*$]{\vtop{%
  \vskip0pt
  \hbox{\includegraphics[width=0.48\textwidth]{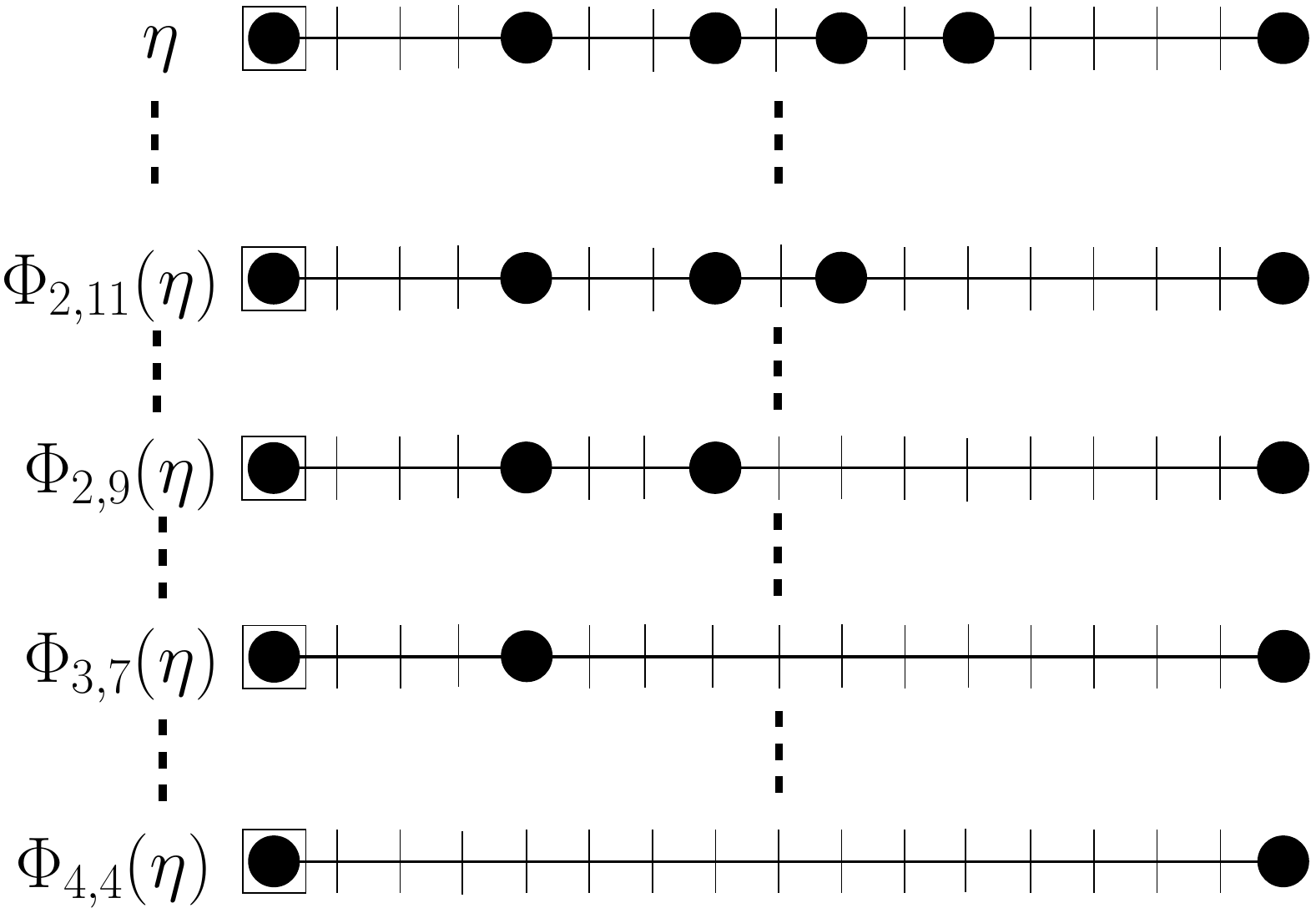}}}}\quad \quad 
    \subfloat[$\eta^z \notin A_*$]{\vtop{%
  \vskip0pt
  \hbox{\includegraphics[width=0.48\textwidth]{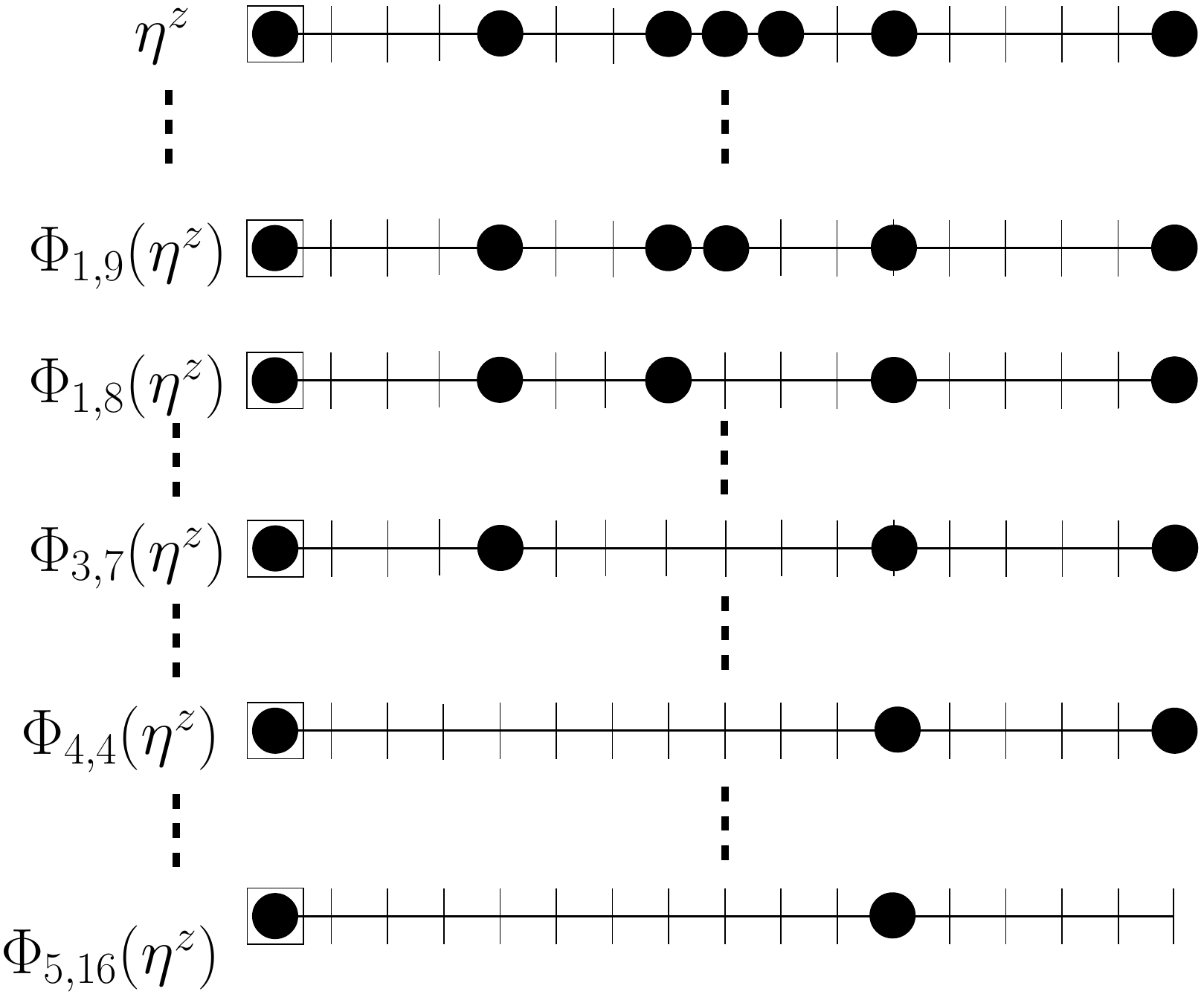}}}}}
  \caption{\label{fig:det-dyn} Example of the deterministic dynamics on a
    system of size $L = 16$. 
    Zeros are indicated by solid circles and ones by $|$, the frozen zero at
    the origin is contained in a box. 
    The deterministic dynamics $\Phi_{d,x}(\eta)$ are
    defined in Section \ref{Astar} (see Remark \ref{africano}), this algorithm acts by only removing zeros which are
    intially present in $\eta$ in order of their gap size and never adding zeros.
    Only the steps of the deterministic dynamics algorithm which change the configuration are shown. 
    Left: An example of the deterministic dynamics applied to a configuration $\eta \in \partial A_*$. $\dyn{\eta}{L-1}{1} = \10$ so $\eta \in A_*$ (see Definition \ref{def:Astar}). 
    Right: The dynamics applied to $\eta^z$ with $z=8$, the zero at $L$ is removed before all other zeros are removed so that $\eta^z\notin A_*$.}
\end{figure}

\subsubsection{Key properties of $A_*$ and proof of the lower bound in Theorem \ref{paletti}} Our key result is the following:
\begin{proposition}
\label{collefiorito} 
Let $2 \leq L \leq d/q$ and set $n := \lceil \log _2 L \rceil$.  Then
\begin{equation}
\cD_\L (\mathds{1}_{A_*}) \leq  \frac{q^n 2^{\binom{n}{2} }}{n!}  q^{-\a}
\end{equation}
for some positive  constant $\a$ depending only on $d$. 
\end{proposition}
The following consequence will be useful to prove Theorem~\ref{hetero}.
\begin{corollary}
\label{boundarymeasure}
Under the same assumptions of Proposition \ref{collefiorito}
\[
\pi(\partial A_*)\le   \frac{q^n 2^{\binom{n}{2} }}{n!}  q^{-(1+\a)}.
\]
\end{corollary}
\begin{proof}[Proof of the Corollary]
It follows immediately from Proposition \ref{collefiorito} together
with \eqref{rane} if we observe that the escape rate
$K(\eta,A_*^c)\ge q$ for all $\eta\in \partial A_*$ (cf. \eqref{atmosfera}).  
\end{proof}
We have now all the ingredients to prove the lower bound in Theorem \ref{paletti}.
As already observed  $\mathds{1}0 \in A_*$ and $\{ \eta_L=1\} \subset
A_*^c$, so $\p(A_*)\geq p^{L-1}q$ and $\p(A_*^c) \geq p$. At this point the lower bound on the relaxation time follows at once from \eqref{biancaneve} together with Proposition \ref{collefiorito}. The proof of Theorem \ref{paletti} is now complete modulo the  proof of Proposition \ref{collefiorito} which is the subject of the following section. 
\qed

\subsubsection{Proof of Proposition \ref{collefiorito}} 

We first collect some properties which follow immediately from the definition of
the deterministic dynamics:
\begin{itemize}
\item[(P1)] The deterministic dynamics can only remove vacancies, hence  gaps
  are increasing under the dynamics. 
  Also, if $\eta$ has a vacancy at $x$ with gap $d$, then such a vacancy is still present
  in the configurations $\dyn{\eta}{d'}{x'}$ with $(d',x')\prec  (d,x)$.
\item[(P2)]
  $\dyn{\eta}{d}{x}$ contains no vacancies with gaps smaller than $d$ and all vacancies right of $x$ (including $x$)
  have gaps no smaller than $d+1$.
\item[(P3)] Whether the deterministic dynamics removes
  a vacancy or not at a site $x$ only depends on the configuration to
  the left of $x$,  so for all $d,x,y \in [1,L]$   the value  $\dyn{\eta}{d}{y}_x$ is
  independent of $\eta_{(x,L]}$.
  
\item[(P4)]
  Once the deterministic dynamics started from two different initial
  configurations are the same, then they remain the same. 
  Equivalently, for two configurations $\eta$ and $\eta'$ if there exists $(d,x)$ such that
  $\dyn{\eta}{d}{x}= \dyn{\eta'}{d}{x}$ then $\dyn{\eta}{\bar d}{\bar x} =
  \dyn{\eta'}{\bar d}{\bar x }$ for all $(\bar d, \bar x ) \succ (d,x)$. 
  In this case we say that the two dynamics, starting from $\eta$ and $\eta'$ respectively, couple.
\end{itemize}

In what follows, we set $\eta_0:=0$ to denote the frozen zero at the origin.
To simplify the notation we continue to write $\L$ for $[1,L]$ and introduce  $\L^0=[0,L]$ (note $\L^0$ includes the origin on which there will always be a fixed vacancy).

\medskip

\begin{lemma}
\label{gallo2}
Let $\eta \in \partial A_*$ and let $z \in\L$ be such that $c^\L_z( \eta)=1$ and $\eta^z \not \in A_*$.
Take an interval $I=[a,b]\subseteq \L^0$ containing $z$ and $z-1$. 
Define $\ell := b-a = |I| - 1$ and 
\begin{align}
    I_-:= (a-\ell,a)\cap \L^0,\nonumber\\
    I_+ := (b, b+\ell]\cap \L^0.
\end{align}
Then $I = \L^0$ or $\eta$ has at least one vacancy in $ (I_- \cup I_+) \subset  \L^0$.
\end{lemma}
\begin{remark}
  If $\eta \in \partial A_* \subset A_*$ then $\eta_L = 0$ since vacancies can only be removed by the deterministic dynamics.
  If $I \neq \L^0$ then $(I_- \cup I_+) \subset  \L^0$ may contain the origin or the site $L$ on which $\eta$ is necessarily zero for all $\eta  \in\partial A_*$. 
  This case is not excluded from the lemma. 
  Fig. \ref{fig:gallo} shows an illustrated application of the lemma.
\end{remark}

\begin{figure}[t]
  \centering
  \mbox{\hbox{\includegraphics[width=0.75\textwidth]{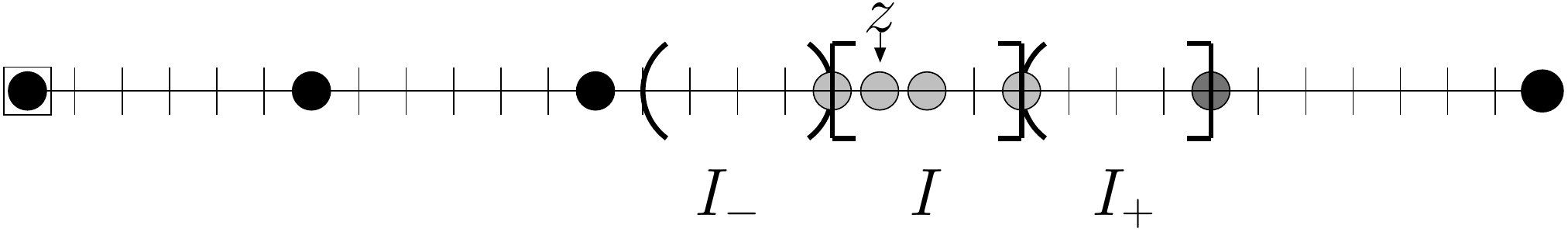}}}
  \caption{\label{fig:gallo} Lemma \ref{gallo2} applied to a configuration on a system of size $L = 32$. Zeros are shown by circles and ones by $|$ . The frozen zero at the origin is contained in a box. A configuration $\eta \in \partial A_*$ is shown where $\eta^z\notin A_*$ for $z = 18$. Lemma \ref{gallo2} is illustrated with $I =[z-1,z+3]$, zeros of $\eta$ contained in $I$ are light grey, $\eta$ does not contain any zeros on $I_-$ and there is a single vacancy in $I_+$ coloured dark grey.}
\end{figure}

\begin{proof}  
Fix $\eta \in \partial A_*$ and let $z \in\L$ be such that $c^\L_z( \eta)=1$ and $\eta^z \not \in A_*$.
We first note that $z-1 \in I$ and  $\eta_{z-1} = 0$ because $c^\L_z( \eta)=1$. 
Also $z < L$, otherwise $\dyn{\eta}{1}{L}_L = 1$ which contradicts $\eta \in A_*$.
Suppose for contradiction that $I \neq \L^0$ and $\eta_y = 1$ for all $y \in I_{-} \cup I_+$. 
Since $I \neq \L^0$ we have $\ell < L$. We will prove that 
\begin{align}\label{porta}
  \dyn{\eta}{\ell}{a+1}_y &= \dyn{\eta^z}{\ell}{a+1}_y \qquad \forall y \in \L_{L  },
\end{align}
which leads to a contradiction by property (P4) and the definition of $A_*$.  
Observe that $I \neq \L^0$ implies $(I_{-} \cup I_+) \neq \emptyset$ since $1\leq \ell < L$ and $z<L$.
\medskip

{\bf Claim: \eqref{porta} holds for all $y\in \L$ with $y <z$}. 
This follows immediately from property (P3) and the fact that $\eta$ and $\eta^z$ coincide on $[1, z)$. 

\medskip

{\bf Claim:  \eqref{porta} holds for all $y\in\L$ with $y \geq b+1$}. 
Our supposition, $\eta_y =1$ for all $y \in I_-\cup I_+$, implies $b+\ell < L$ since $\eta$ necessarily has a vacancy at $L$ for each $\eta \in A$. 
Since $z\leq b$, $\eta$ and $\eta^z$ coincide on $[b+1, L]$. 
Since $\eta_L = 0$ there is at least one vacancy on $[b+1,L]$, call $u$ the position of the leftmost vacancy of $\eta$ on $[b+1,L]$.
By hypothesis $\eta_y=1$ for all  $y \in [b+1, b+\ell]$ so the vacancy at $u$ has gap  $g_u(\eta)>\ell$  and $g_u(\eta^z)>\ell$.
By property (P1) the vacancy at $u$ is still present in $\dyn{\eta}{\ell}{a+1}$, $ \dyn{\eta^z}{\ell}{a+1}$. 
Since $\eta$ and $\eta^z$ coincide on the right of $u$, it is trivial to check that the deterministic dynamics starting from $\eta$ and $\eta^z$ coincides on $[u,L]$ until the vacancy at $u$ is removed. 
This proves \eqref{porta} for $y \geq u$. 
Since $\eta$ and $\eta^z$ have no vacancy in $[b+1, u)$, from property (P1) we derive \eqref{porta} for $ y \in [b+1,u)$. 
This concludes the proof of our claim.

\medskip

{\bf Claim: \eqref{porta} holds for all $y\in\L$ with $z\leq y \leq b $}. 
We define the gap of the frozen vacancy at the origin as infinite, since it remains under the dynamics for all times, $g_0(\eta) = g_0(\eta^z) := \infty$.
Let $u$ be the leftmost zero of $\eta$ and $\eta^z$ that is contained in $I$, $u$ is well defined since $\eta_{z-1}=\eta^z_{z-1} =0$ and $\eta$ and $\eta^z$ coincide outside of $z$.
By assumption $g_u(\eta) = g_u(\eta^z) \geq \ell$.  
Due to (P1) and (P2) we conclude that  $\dyn{\eta}{\ell}{u+1}$ and $\dyn{\eta^z}{\ell}{u+1}$ still have a vacancy at $u$. 
This fact implies that both $\dyn{\eta}{\ell}{u+1}$ and $\dyn{\eta^z}{\ell}{u+1}$ have no vacancy in $[u+1, b]$, otherwise they would have some vacancy of gap at most $b-u\leq\ell$ in contradiction with (P2).  
\end{proof}

The above Lemma \ref{gallo2} allows us to isolate a special subset of vacancies for a generic configuration $\eta \in \partial A_*$. 
This special subset will be defined iteratively. 
To this aim we first associate to $\eta \in \partial A_*$ an increasing  family of subsets 
$\D_1\subset \D_2 \subset \cdots \subset \D_K$ in $\L^0$, where $\D_i$ contains at least $i$ vacancies, by the  algorithm  described below.

For clarity we first describe the algorithm in words. 
Take $ \eta \in \partial A_*$ and define $z_0$ by choosing a site $z \in \L$ such that $c_z^\L(\eta)=1$ and $\eta^z \not \in A_*$. 
We know that $\eta_{z_0-1}=0$ since $c_z^\L(\eta)=1$, set $\D_1= [z_0-1,z_0]$. 
Let $a:= z_0-1$, $b:=z_0$, suppose $L>1$ so that $\D_1\neq \L^0$, we know by Lemma \ref{gallo2} that there is a vacancy of $\eta$ in $I_-\cup I_+\subset \L^0$  (where $I_-$ and $I_+$ are defined as in  Lemma \ref{gallo2}), and since $\ell = 1$ we have $I_-=\emptyset$ and $I_+ = \{z_0+1\}$.
Let $x_1 = z_0+1$ and make $\D_2$ the extension of $\D_1$ to include $x_1$, $\D_2:= [a,x_1]$.
We now proceed by induction, defining $I = \D_k$  then $\D_{k+1}$ is the extension of $\D_k$ having $x_k$ as extreme, where $x_k$ is a vacancy of $\eta$ in $I_-\cup I_+$, by applying Lemma \ref{gallo2}, until  $\D_K =  \L^0$ (if there is more than one vacancy, fix a rule to specify $x_k$ uniquely).
At a certain moment we will cover $\L^0$, that is we arrive at a set  $\D_K$ such that $\D_K = \L^0$, and the algorithm will stop.
In this way we show that $\eta$ contains at least a certain number of
vacancies which must also satisfy certain geometric constraints. 

We now make the algorithm described above precise. 
It is convenient to use the following notation: given the interval  $I=[a,b]\subset \bbN $  and a site  $x \in \bbN \setminus I$, we define $I \star x$ as the set $[a,x]$ if $x >b$ and as the set $[x,b]$ if $x<a$. We assume that $L>1$.
The input of the algorithm is given by the pair $(\eta,z_0)$ where $\h
\in \partial A_*$ and $z_0 \in \L$ is such that $c^\L_{z_0}(\eta)=1$ and $\eta^{z_0} \not \in A_*$.

\medskip

\centerline{
\emph{Algorithm to  determine $K, \D_1, \D_2, \dots, \D_K$ given  $(\h,z_0)$.}
}

\medskip

\begin{itemize}

\item STEP 0: Set $z_1= z_0-1$ and $\D_1:= [z_1,z_0]$.
 
\medskip

\item INDUCTIVE STEP.  Suppose we have defined $\D_1, \D_2, \dots , \D_k$, let $I = \D_k$. Define  $\ell$, $I_-$ and $I_+$ as in Lemma \ref{gallo2}.

\medskip

\begin{itemize}
\item \textbf{Case 1:}  If $\D_k = \L^0$ set $K=k$ and STOP.

\smallskip

\item \textbf{Case 2:} If $\D_k \neq \L^0$ then let $z_{k+1}$ be the position of the vacancy of $\eta$ in $I_- \cup I_+$  which is nearest to the border of $\D_k$ (take the leftmost one if two are of equal distance). Such a vacancy at $z_{k+1}$ exists due to Lemma \ref{gallo2}. Set $\D_{k+1}:= \D_k \star z_{k+1}$.
\end{itemize}
\end{itemize}

\medskip

Since $\D_{k+1}$ is obtained from  $\D_k$ by enlarging it, the above algorithm always stops.
Note that each interval $\D_{k+1}$ is obtained by extending $\D_k$ either on the left or on the right. 
Hence $\D_{k+1}$ has one extreme in common with $\D_k$ and one extreme not belonging to $\D_k$. 
The following observation  is fundamental and follows immediately from the definition of the algorithm (we omit its proof):

\begin{lemma}\label{rumore}
The vacancies of $\h \in \partial A_*$ in $\L^0 \setminus \{z_0\}$ are located at $\{z_1,z_2,\ldots,z_K \}$, moreover
$$  
\left|\L \cap \{z_1, z_2, \dots , z_K\}\right | = K-1\ .
 $$
In particular the number of vacancies of $\eta \in \partial A_*$ is $K-1$ or $K$, if $\eta_{z_0} = 1$ or $0$ respectively.
\end{lemma}
We recall that the set $\{z_1, z_2, \dots , z_K\}$ depends on $\eta \in \partial A_*$, although it certainly contains the origin, on which there is a frozen vacancy, and $L$ since  $\eta_L = 0$ for each $\eta \in \partial A_* \subset A_*$. 

We now isolate some geometric properties of $z_0,z_1,z_2, \dots z_K$.  
First note that given $z_0,z_1, \dots ,z_K$ we can recover $\D_1, \D_2, \dots, \D_K$.
Given $z_0$ the first two positions of vacancies $z_1$, $z_2 $ are determined by Lemma \ref{gallo2} ($z_1=z_0-1$ and $z_2=z_0+1$). 
To describe the points $z_0,z_1,z_2,z_3,\dots , z_K$ we can use the following formalism. 
For each $k: 2\leq k \leq K $ we set $\e_k=-1$ if $z_k$ is on the left of $\D_{k-1}$  and $\e_k=+1$ otherwise; while we define $d_k$ as the Euclidean distance of $z_k$ from $\D_{k-1}$.  
Hence, writing $\D_{k-1}=[a,b]$, we have $z_{k}= a-d_k$ if $\e_k=-1$ and $z_k= b+d_k$ if $\e_k=+1$.  
Writing $\ell_k$ for the length of $\D_k$ ($\ell_k = |\D_k| - 1$),  note that
\begin{equation}
\begin{cases}
\ell_1= 1\,, & \\
d_k \leq \ell_{k-1}    & \forall k: 2\leq k \leq  K \,,\\
\ell_k = \ell_{k-1}+d_k & \forall k: 2\leq k \leq  K \,.
\end{cases}
\end{equation}
This implies that  $\ell_k=  d_1+d_2+d_3+\cdots+d_k$ (where $d_1:=1$) and $ \ell_k = \ell_{k-1}+d_k \leq 2 \ell_{k-1}$ for $ 2\leq k \leq K$. 
In particular, 
\begin{equation}\label{palla}
\ell_k \leq 2^{k-1} \qquad \forall k \,:\,1\leq k \leq K\,.
\end{equation}
When the algorithm stops $\D_K = \L^0$ and $\ell_k = L$, so $L \leq 2^{K-1}$ which implies $K \geq n+1 $ where  $n = \lceil \log _2 L\rceil$.

Due to \eqref{rane} and \eqref{atmosfera} we have
\begin{align}
  \label{eq:D-new}
  \cD_\L(\mathds{1}_{A_*})=  \sum_{z_0 \in \L}\left[p\, \pi\left(\partial A_*^{z_0,0}\right) + q\, \pi\left(\partial A_*^{z_0,1}\right)\right]\ ,
\end{align}
where
\begin{align*}
  \partial A_*^{z_0,0}&:=\{\eta \in \partial A_* \sst c_{z_o}^\L(\eta) = 1,\ \eta_{z_0} = 0 \text{ and } \eta^{z_0}\notin A_*\}\,,\\
\partial A_*^{z_0,1}&:=\{\eta \in \partial A_* \sst c_{z_o}^\L(\eta) = 1,\ \eta_{z_0} = 1 \text{ and } \eta^{z_0}\notin A_*\}\,.
\end{align*}
Collecting all the above geometric considerations we have the following result:
\begin{lemma}
The boundary $\partial A_*$ satisfies $\partial A_* = \bigcup_{z_0\in\L}\bigcup_{i\in\{0,1\}}\partial A_*^{z_0,i} $ and
\begin{align}
  \label{eq:A-decomp}
\partial A_*^{z_0,i} \subset \left\{ \eta \in \O_\L \,:\, \eta_{z_0} = i \text{ and } \eta_z =0 \; \forall z \in W \cap \L \text{ for some } W \in \G_{z_0} \right\}
\end{align}
where the set $\G_{z_0} $ is  given by the families $ \{z_1, \dots, z_{n+1} \}$ of distinct points in $\L^0$ such that
\begin{itemize}
\item the position of $z_1=z_0-1$ is uniquely determined by  $z_0 \in \L$,
\item $| \{z_1, \dots, z_{n+1} \} \cap \L | \geq n$,
\item there exist positive integers $d_1, \dots, d_{n+1}$ such that
\begin{equation}\label{paride}
 \begin{cases}
 d_1 = 1\,, & \\
 d_k \leq d_1+d_2+ \cdots +d_{k-1} \qquad &
\forall k: 2\leq k \leq n+1\,.
\end{cases}
\end{equation}
\item there exist $\e_2, \cdots, \e_{n+1}  \in \{-1,+1\}$ such that, setting $\D_1:=[z_1,z_0]$, the following recursive identities are satisfied for $k=2,\dots, n+1$
\begin{equation}
\begin{cases}
 z_k= a- d_k & \text{ if } \e_k=-1\,,\\
 z_k= b+d_k & \text{ if } \e_k=1\,,
 \end{cases}
 \end{equation}
where $ \D_{k-1}=[a,b]$ and $\D_k := \D_{k-1} \star z_k$.
\end{itemize}
\end{lemma}

As immediate consequence of the above lemma we get, for $i=1$ or $0$
\begin{equation}
\label{chiave}
 q^{i}\p( \partial A_*^{z_0,i} ) \leq q^i \sum _{W \in \G_{z_0}} \p( \eta_{z_0} = i \text{ and } \eta_z =0 \; \forall z \in W \cap \L) \leq q^{n+1} |\G_{z_0}|\,.
 \end{equation}
To estimate $|\G_{z_0}|$ we use the following result:
\begin{lemma}
The numbers of strings $(d_2,d_3, \dots, d_{n+1})$ of positive integers satisfying \eqref{paride} is bounded from above by $\frac{2^{ \binom{n}{2}}}{n!}$, hence
\begin{align}
  \label{eq:G-bound}
  |\G_{z_0}| \leq \frac{2^{n} 2^{ \binom{n}{2}} }{n!}\,.
\end{align}
\end{lemma}
\begin{proof}
We give an iterative bound.  
Given an integer $j$ we define 
$$U(j):=\{ (x_1, x_2, \dots, x_j)  \,:\, 0 \leq x_1 \leq 1 \text{ and } 0 \leq  x_k \leq  x_1+ \dots+ x_{k-1}  \; \forall k: 2\leq k \leq j\}\ .$$
Setting $M_j:=x_1+ x_2 +x_3 + \cdots+ x_j$, integrating on the last variable we get
 \begin{equation}\label{elena1}
\int _{U(j) } dx_1 dx_2 \cdots d x_{j} = \int _{U(j-1) } dx_1 dx_2  \cdots dx_{j-1} M_{j-1}.
\end{equation}
We now prove by induction that 
\begin{equation}\label{elena2}
\int_{ U(j-m) }dx_1dx_2  \cdots d x_{j-m} M_{j-m} ^m \leq \frac{2^{m+1}}{m+1} \int_{ U(j-m-1) }dx_1 dx_2 \cdots d x_{j-m-1} M_{j-m-1} ^{m+1} \,,\end{equation}
for all $m\geq 1$ and $ j-m-1 \geq 1$. 
By integrating on the last variable we get 
\begin{align*}
\int_{ U(j-m) }dx_1 dx_2 &\cdots d x_{j-m} M_{j-m}^m =  \int_{ U(j-m) }dx_1 dx_2 \cdots d x_{j-m} \left( M_{j-m-1}+x_{j-m}\right) ^m\\
&=  \int_{ U(j-m-1) }dx_1 dx_2 \cdots d x_{j-m-1}\left[{\frac{ \left( M_{j-m-1}+x_{j-m}\right) ^{m+1} }{m+1}}\right]_{x_{j-m}=0}^{ x_{j-m} = M_{j-m-1} }\\
&\leq \frac{2^{m+1}}{m+1} \int_{ U(j-m-1) }dx_1 dx_2 \cdots d x_{j-m-1} M_{j-m-1} ^{m+1}\,.
\end{align*}
Combining \eqref{elena1} and \eqref{elena2} we conclude that
\begin{align*}
\int _{U(j) } dx_1 dx_2 \cdots d x_{j} &\leq \frac{2^2\cdot 2^3\cdots 2^{j-1} }{2\cdot 3\cdots (j-1)} \int _{U(1)}dx_1 M_1^{j-1}\\
&\leq  \frac{2^{ \binom{j-1}{2}} }{(j-1)!}\,.
\end{align*}
The result then follows from this bound by observing that the number of strings we want to estimate is bounded from above by
$\int _{U(n+1)} dx_1 d x_2 \dots d x_{n+1}$, and observing that there at most $2^{n}$ ways to choose $\e_2, \dots, \e_{n+1}$. 
\end{proof}

Combining \eqref{eq:D-new},~\eqref{eq:A-decomp},~\eqref{eq:G-bound}, and observing that there are $L$ choices for $z_0$ we find 
\begin{equation}
\cD_\L  \left( \mathds{1}_{A_*}\right)  \leq q^{n+1 }  \frac{2^{n+1}L\,2^{\binom{n}{2} }}{n!}\,.
\end{equation} 
thus implying Proposition \ref{collefiorito} (recall that $L \leq d/q$).
Since $\10 \in A$ and $\{\eta_L = 1\} \in A^c$ we have $\pi(A) \geq p^{L-1}q$ and $\pi(A^c) \geq p$, so \eqref{biancaneve} gives rise to the lower bound
\begin{align*}
  T_{\rm rel} (L)  \geq \frac{n!}{q^{n}2^{\binom{n}{2} }}\frac{p^L}{2^{n+1}L}\, \quad \textrm{where} \quad n = \lceil \log_2 L \rceil\,.
\end{align*}




%
%
%

\section{Time scale separation and dynamic heterogeneity: proofs}
\subsection{Proof of Theorem \ref{bubu} }
\label{proofbubu}
Point (i) of Theorem \ref{bubu} is an immediate consequence of the
bound  \eqref{sciroppo}. The time scale separation expressed by
\eqref{fabio1} is a corollary of \eqref{piscina}, as follows. 
Let $L=d/q^\gamma$ with $\g \in (0,1)$, $d>0$, and $L'=\l L$   with $ \lambda > 1$.
Then $L,L'\in [1,1/q]$ for $q$ sufficiently small,
so Theorem  \ref{paletti} implies there exists some universal constant $\bar\a>0$ such that,
\begin{align}
  \label{eq:Tratio}
  \frac{T_{\rm rel}(L')}{T_{\rm rel}(L)} \geq \frac{n'!2^{\binom{n}{2}}}{n!2^{\binom{n'}{2}}}q^{(n-n') + \bar\a}, \quad n = \lceil \log_2 L \rceil, \ \ n'=\lceil \log_2 L' \rceil\,.
\end{align}
Let $k := (n'{-}n) \geq \lfloor \log_2 \l \rfloor$ which is independent of $q,d$ and $\g$.
Since $ \frac{d}{q^\g} \leq 2^n \leq \frac{2d}{q^\g}$,
the above bound  \eqref{eq:Tratio}  and some straightforward algebra lead to,
\begin{align*}
  \frac{T_{\rm rel}(L')}{T_{\rm rel}(L)} \geq C n^k
  q^{\bar\a - \left(1 - \g\right)k}
\end{align*}
for some $C$ independent of $q$.
The result follows by choosing $\l$ large enough.

It remains to show that, for $\g<1/2$ and $L=
d/q^\g$, $T_{\rm rel}(2L)\succ T_{\rm rel}(L)$. 
For this purpose define $L':=2L$ and set $t= q^\b T_{\rm
  rel}(L)$ with $0<\b <1-\g$. 
A union bound shows that, for any integer $N$, 
\begin{gather}
\bbP^{\L'}_{\mathds{1}{0}}(\eta_x(s)= 1\ \forall x\neq 2L, \ \forall s=i t,\ i=1,\dots
N)\nonumber\\
\ge 1-\sum_{i=1}^N \sum_{x=1}^{2L-1}\frac{1}{\pi_{\L'\setminus \{2L\}}(\1)}\bbP^{\L'}_\pi(\eta_x(it)=0)
\nonumber\\
\ge 1-\frac{2dN}{(1-q)^{2L-1}}q^{1-\g}=1-2dN q^{1-\g}(1+o(1)),
\label{fabiodom}
\end{gather}
where $\L'=\{1,2,\dots, L'\}$.
Consider now the East model in $\L'$ starting from $\mathds{1}{0}\in\O_{\L'}$  and
let $\t_L$ be the first time that there is a legal ring at $x=L$ with
the corresponding coin toss equal to one. Clearly $\t_L$ has the same
law as the hitting time $\t_{\eta_L=1}$ under the measure
$\bbP^\L_{\mathds{1}{0}}$. Define the auxiliary
Markov time $\tilde \tau$ by 
\[
\tilde \tau =\inf\{s>\t_L: \eta(s) =\mathds{1}{0} \ \textrm{or}\ \ \eta_{2L}(s) = 1\}.
\] 
Writing 
\begin{align*}
  \bbP_{\mathds{1}{0}}^{\L'}(\tilde\tau-\tau_L\ge t)\le 
\bbP_{\mathds{1}{0}}^{\L'}\left(\t_L\ge q^{-\epsilon}\,T_{\rm rel}(L)\right)+ 
\bbP_{\mathds{1}{0}}^{\L'}\left(\tilde\tau-\tau_L\ge t\ ;\ \t_L\le q^{-\epsilon}\,T_{\rm rel}(L)\right)\,,\nonumber
\end{align*}
we can bound the first probability on the right hand side by $O(q^\epsilon)$  using Proposition \ref{dominare}.
We may bound the second probability by using \eqref{fabiodom} with $N=\lceil q^{-\epsilon}T_{\rm rel}(L)/t\rceil$ and $\epsilon= (1-\g-\b)/2$, 
so that
\begin{align}
  \label{fdom1}
\bbP_{\mathds{1}{0}}^{\L'}(\tilde\tau-\tau_L\ge t) = O(q^{(1-\g-\b)/2}).
\end{align}

\begin{claim} 
\label{claim6}
For any $\g< 1/2$,
\[
\bbP_{\mathds{1}{0}}^{\L'}(\eta(\tilde \tau)=\mathds{1}{0})= 1-O(q^\d)
\]  
for some $\d>0$.
\end{claim}
\begin{proof}[Proof of the Claim]
We write 
\begin{equation}
  \label{fdom2}
 \bbP_{\mathds{1}{0}}^{\L'}(\eta_{2L}(\tilde \tau)=1)   \le \bbP_{\mathds{1}{0}}^{\L'}(\eta_{2L}(\tilde \tau)=1;\ \tilde \tau{-}\tau_L\le t) +
\bbP_{\mathds{1}{0}}^{\L'}(\eta_{2L}(\tilde \tau)=1 ;\ \tilde \tau{-}\tau_L\ge t)
\end{equation}
 The last term in the r.h.s. of \eqref{fdom2} is $O(q^{(1-\g-\b)/2})$ because of \eqref{fdom1}.
Let us examine the first term. The strong Markov property gives
\begin{gather*}
\bbP_{\mathds{1}{0}}^{\L'}(\eta_{2L}(\tilde \tau)=1\ ;\ \tilde \tau-\tau_L\le t)
\le \bbP_{\mathds{1}{0}}^{\L'}(\tau_{\eta_{2L}=1}-\tau_L\le t)\\
\le \max_{\eta\in \O_{L-1,2L}}\bbP_\eta^{\L'}(\tau_{\eta_{2L}=1}\le t).
\end{gather*}
where $\O_{L-1,2L}=\{\eta\in \O_{\L'}:\  \ \eta_{L-1}=\eta_{2L}=0,\ \eta_x=1 \, \forall x\in [L,2L-1]\}$.

Choose $\eta\in \O_{L-1,2L}$ and declare the vacancy at $x=L-1$
to be the \emph{distinguished zero} at time zero (see
e.g. \cite{Aldous} or \cite{CMST}). At any later time $s>0$ the
position $\xi(s)$
  of the distinguished zero is determined according to the following iterative
  rule: \\
(i) If $\xi(s)>0$ then  $\xi(s')=\xi(s)$ for all times $s'>s$ which
are strictly smaller than the time $s_1$ of the first
  \emph{legal} ring  at $\xi(s)$;\\
(ii) at time $s_1$  the distinguished zero $\xi(s)$ jumps to $\xi(s)-1$; \\
(iii) if $\xi(s)=0$ then $\xi(s')=0$ for all $s'>s$.\\
Thus, with probability one, the path $\{\xi(s)\}_{s\le t}$ is
right-continuous, piecewise
constant, non increasing, with possibly $n\in \{0,1,\dots, L-1\}$
discontinuities at times $s_1<s_2 <\dots <s_n$ at which it decreases by one.  
In the sequel we will
adopt the standard notation $\xi_{s^{-}}:=\lim_{\delta\uparrow  0}\xi_{s-\d}$, and set $s_0:=0$. 
\begin{remark}
Because of the orientation of the East constraint, fixing the path $\{\xi(s)\}_{s\le t}$ has no influence whatsoever on
the Poisson rings and coin tosses to the right of the path itself. Thus the evolution to the right of the path
$\{\xi_s\}_{s\le t}$ is still an East evolution in a domain whose left
boundary jumps by one site to the left at the jumps of the path.
\end{remark}
The key property of the distinguished zero is the
following \cite{Aldous}*{Lemma 4}. Suppose that the configuration
$\eta$ to the right of
$L-1$ and to the left of $2L$ was chosen according to the reversible measure $\pi$ (instead of
being identically equal to $1$). Then, at any
given time $s>0$ and conditionally on the path
$\{\xi(s)\}_{s'\le s}$, the law of the restriction of $\eta(s)$ to the
interval $\{\xi(s)+1,\dots,2L-1\}$ is \emph{again} $\pi$.

Using the same argument leading to \eqref{fabiodom} together with the
above property, if  
\[
\O_s=\bigl\{\eta\in \O_{\L'}:\ \eta_x=1 \ \forall x\in
[\xi(s)+1,2L-1]\bigr\},
\] 
then
\[
\bbP^{\L'}_\eta\left(\exists\, i\le n: \ \eta(s_i)\notin \O_{s_i}\tc
  \{\xi_s\}_{s\le t}\right)=O(nq^{1-\g})=O(q^{1-2\g}).
\]
As a consequence 
\begin{align*}
  \bbP^{\L'}_\eta(\tau_{\eta_{2L}=1}\le t\tc \{\xi_s\}_{s\le t})
&\le 
\sum_{i=0}^n  \bbP^{\L'}_\eta\left(\tau_{\eta_{2L}=1}\in
  (s_i,s_{i+1})\tc \eta(s_i)\in \O_{s_i}\,;\,\{\xi_s\}_{s\le t}\right)
\\
&+ \bbP^{\L'}_\eta\left(\tau_{\eta_{2L}=1}\in (s_n,t]\tc \eta(s_n)\in \O_{s_n}\,;\,\{\xi_s\}_{s\le
  t}\right) 
+ O(q^{1-2\g}).
\end{align*}
Let us examine a generic term $\bbP^{\L'}_\eta\left(\tau_{\eta_{2L}=1}\in
  (s_i,s_{i+1})\tc \eta(s_i)\in \O_{s_i}\,;\,\{\xi_s\}_{s\le
    t}\right)$. Since (i) the distinguished zero does not move in the time
interval $[s_i,s_{i+1})$, (ii) the interval $\{\xi(s_i)+1,\dots, 2L\}$ has
length at least $L$ and (iii) $\eta(s_i)\in \O_{s_i}$, the above
probability is smaller than
$\bbP_{\mathds{1}{0}}^\L(\tau_{\eta_{L}=1}\le s_{i+1}-s_i)$. 
Due to Proposition \ref{dominare} together with Theorem \ref{equo}  we have
\[
\bbP_{\mathds{1}{0}}^\L(\tau_{\eta_{L}=1}\le s_{i+1}-s_i)\le
e(s_{i+1}-s_i)/T_{\rm hit}(L)\le c (s_{i+1}-s_i)/T_{\rm rel}(L) .
\]
Thus 
\begin{gather*}
\sum_{i=0}^n  \bbP^{\L'}_\eta\left(\tau_{\eta_{2L}=1}\in
  (s_i,s_{i+1})\tc \eta(s_i)\in \O_{s_i}\,;\,\{\xi_s\}_{s\le t}\right)\\
+ \bbP^{\L'}_\eta\left(\tau_{\eta_{2L}=1}\in (s_n,t]\tc \eta(s_n)\in \O_{s_i}\,;\,\{\xi_s\}_{s\le
  t}\right) \\
\le c t /T_{\rm rel}(L)=O(q^\beta)
  \end{gather*}
In conclusion 
\[
\max_{\eta\in \O_{L-1,L}}\bbP_\eta(\tau_{\eta_{2L}=1}\le t)=
O(q^\beta)+ O(q^{1-2\g})
\]
and the claim follows with $\d= \min\left(\b,1-2\g,(1-\g-\b)/2\right)$.
\end{proof}
Back to the proof of $T_{\rm rel}(2L)\succ T_{\rm rel}(L)$ we observe that, on the event
$\{\eta_{\tilde \tau}=\mathds{1}{0}\}$, the hitting time $\tau_{\eta_{2L}=1}$ is larger than
$\tau_L+\tau'$, where $\tau'$ is distributed as $\tau_{\eta_{2L}=1}$
and it is independent of $\tilde \tau$. Hence, using Claim
\ref{claim6} and  Proposition \ref{dominare},
\begin{align*}
T_{\rm hit}(2L)&=\bbE^{\L'}_{\mathds{1}{0}}(\tau_{\eta_{2L}=1})\ge 
\bbE^{\L'}_{\mathds{1}{0}}(\tau_{\eta_{2L}=1}\mathds{1}_{\eta(\tilde \tau)=\mathds{1}{0}})
\\
&\ge 
\bbE^{\L'}_{\mathds{1}{0}}(\tau_{L}\mathds{1}_{\eta(\tilde
    \tau)=\mathds{1}{0}})+
\bbE^{\L'}_{\mathds{1}{0}}(\tau'\mathds{1}_{\eta(\tilde
    \tau)=\mathds{1}{0}})
\\
&\ge T(L) \bbP^{\L'}_{\mathds{1}{0}}(\tau_{L}\ge T(L)\,;\,\eta(\tilde
    \tau)=\mathds{1}{0})                    + (1-O(q^\d)) T_{\rm hit}(2L)
\\
&\ge T(L)\left[1/4-    \bbP^{\L'}_{\mathds{1}{0}}(\eta(\tilde
    \tau)=\mathds{1}) \right] + (1-O(q^\d)) T_{\rm hit}(2L) \\
&\ge T(L)\left[1/4- O(q^\d)) \right] + (1-O(q^\d)) T_{\rm hit}(2L)
\end{align*}
which implies 
\[
T_{\rm hit}(2L)\ge c q^{-\d}T(L).
\]
Here $T(L)$ is such that $\bbP_{\10}^\L(\tau_L\ge T(L))=1/4$. 
Using that $T(L)\asymp T_{\rm hit}(L)\asymp T_{\rm rel}(L)$ we conclude
the proof.

%
%
%
%

\subsection{Proof of Theorem \ref{sorpresa}} 
We prove \eqref{salvia1} then \eqref{salvia2} is a trivial consequence. It
is enough to compare the scale $d/q$ with $1/q$ (recall that  we write
$d/q$ instead of $\lfloor d/q \rfloor$). In particular, if for any
$\d>0$ we can show that
\begin{align}
& T_{\rm rel}(1/q) \leq T_{\rm rel} (d/q) \leq C \,T_{\rm rel}(1/q) \,, \qquad \forall d \in [1,1/\d]\,, \label{pineta1}\\
& T_{\rm rel}(d/q) \leq T_{\rm rel} (1/q) \leq C' \,T_{\rm rel}(d/q) \,, \qquad \forall d \in [\d,1 ]\,.\label{pineta2}
\end{align}
for suitable constants $C,C'$ depending only $\d$, then we immediate
get \eqref{salvia1}. 
Notice that the first bound in \eqref{pineta1} and the second bound in
\eqref{pineta2} trivially follow from the monotonicity of the
relaxation time w.r.t. the interval length (see Lemma \ref{monotone}).

Let us prove that $T_{\rm rel} (d/q) \leq C \,T_{\rm rel}(1/q) $ for
all  $ d \in [1,1/\d]$. To this aim we use the block dynamics as in the proof of the upper bound in Theorem \ref{paletti}.   Given an integer length $\ell\in [1/q , 1/\d q]$ consider the block dynamics  on $[1,\ell]$ in which the left half $\L_1:=[1, \lfloor \ell/2 \rfloor ]$ goes to equilibrium with rate $1$, while the second half $\L_2:=[\lfloor \ell/2 \rfloor+1 , \ell ]$ does the same but only if there is a zero in $\D:=[\lfloor \ell/4\rfloor,
 \lfloor \ell/2 \rfloor ]$.  As proven in \cite{CMRT}[p. 480] this dynamics has spectral gap 
 $$ \l:= 1- \sqrt{ (1-q) ^{|\D|} }\sim 1- e^{- q \ell/8} \sim q \ell /8\geq  1/8\,. $$
 On the other hand, as proven in \cite{CMRT} (see also the proof of the upper bound in Theorem \ref{paletti}), we have
 $$ T_{\rm rel} (\ell) \leq \frac{2}{\l} \max\{ T_{\rm rel}  ( \L_1) ; T_{\rm rel} ( \D \cup \L_2) \}\leq16  T_{\rm rel} ( \lceil (3/4) \ell \rceil)\,. 
  $$
To conclude, one has to apply iteratively the above bound $T_{\rm rel} (\ell) \leq 16  T_{\rm rel} ( \lceil (3/4) \ell \rceil)$ starting from
$ \ell_1 :=  d/q  $ and going from $\ell_i$ to $ \ell_{i+1}:= \lceil
(3/4) \ell _i\rceil$. Clearly the number $m$ of steps necessary to get
to $\ell_m < 1/q$  is $O\left(\ln d / \ln (4/3) \right)$ and it can be
bounded from above by some constant $c(\delta)$. Hence, by the
monotonicity of the relaxation time in the length, we obtain 
$$ T_{\rm rel } (d/q) \leq 16^{m-1}  T_{\rm rel} ( \ell_m) \leq 16^{m-1} T_{\rm rel} (1/q) \,,$$
thus proving our claim \eqref{pineta1}. 
The proof of \eqref{pineta2} is analogous.
\subsection{Proof of Theorem \ref{th:dh}}
\begin{proof}[Proof of (i)]
Without loss of generality we take $d=1$. Let $t:= T_{\rm rel}(1/q^\g)$ and
let also $\epsilon \in (0,1)$ be a small constant to be fixed
later on. 
The same proof of \cite{FMRT-cmp}*{Lemma 4.2} shows that 
\begin{align}
  \label{fabio3bis}
\sup_\eta\bbP_\eta^\L\Bigl(\exists \,z:\
\eta_z(t)=0\textrm{ and }\exists\, s \leq t:\
\eta_z(s)=1\Bigr)=O(q L)=o(1).
\end{align}
Thus 
\begin{gather*}
\sup_\eta\bbP_\eta^\L\Bigl(\exists \,z,z':\ |z-z'|\le \epsilon/q^\g,\ 
\eta_z(t)=\eta_{z'}(t)=0\Bigr)\\=\sup_\eta \bbP_\eta^\L\Bigl(\exists \,z,z':\ |z-z'|\le \epsilon/q^\g,\ 
\eta_z(s)=\eta_{z'}(s)=0\ \forall s\le t\Bigr)+o(1).  
\end{gather*}
Moreover 
\begin{align}
\label{paul1}
\sup_\eta\bbP_\eta^\L&\Bigl(\exists \,z,z':\ |z-z'|\leq \epsilon/q^\g,\ 
\eta_z(s)=\eta_{z'}(s)=0\ \forall s\leq t\Bigr)
\nonumber\\&\le\sum_{\substack{z<z'\\ |z-z'|\le \epsilon/q^\g}}
\sup_\eta\bbP_\eta^\L\Bigl(\eta_{z'}(s)=0\ \forall s\leq t \tc
\eta_{z}(s)=0 \ \forall s\leq t\Bigr)\nonumber \\
&\le \sum_{\substack{z<z'\\ |z-z'|\leq \epsilon/q^\g}}\bbP_{\mathds{1}{0}}^{\L_{z,z'}}(\tau_{\{\sigma_{z'}=1\}}>t)
 \end{align}
where $\L_{z,z'}=[z+1,z']$. Due to Proposition  \ref{dominare}
\[
\bbP_{\mathds{1}{0}}^{\L_{z,z'}}(\tau_{\{\sigma_{z'}=1\}}>t)\le 
e^{-c t/T_{\rm hit}(z'-z)}
\]
for some constant $c$ independent of $z,z'$. If we now combine Theorem 
\ref{equo}, Lemma \ref{monotone} and \eqref{fabio1} we get 
\[
\min_{z'-z\le \epsilon/q^\g}t/T_{\rm hit}(z'-z) \ge c\, T_{\rm rel}(1/q^\g)/T_{\rm
  rel}(\epsilon/q^\g)\ge 1/q^\d
\]
for some $\d>0$ and for all $\epsilon$ small enough. Thus the r.h.s. of \eqref{paul1} is $o(1)$.
\end{proof}

\begin{proof}[Proof of (ii)]
Let
$t:=T_{\rm rel}(\epsilon/q^\g)$ and fix $\eta$ such that: (i)
$\eta_L=0$ and (ii) $L-z\ge 1/q^\g$ where $z:=\max \bigl\{y\in [1,L-1]:\ \eta_y=0\bigr\}$ if the set is non-empty and $z:=
0$ otherwise.  

If $z=0$ then $\eta =\mathds{1}{0}$ and
\begin{gather*}
\bbP_\eta^\L\bigl(\eta_L(t)=0\bigr)\ge\bbP^\L_{\mathds{1}0}(\tau_{\eta_L=1}>t)\\
\ge 1 - et/T_{\rm hit}(L)= 1-o(1)
\end{gather*}
for $\epsilon$ small enough. Above we used Proposition \ref{dominare}
and Theorem \ref{bubu} to bound from
above $t/T_{\rm hit}(L)$.

Assume now $z\neq 0$ and let $\L':=[z+1, L]$.  Let $A_*\subset \O_{\L'}$
be the set given in Definition \ref{def:Astar} with $\L$
replaced by
$\L'$ and $L$ replaced by the cardinality $L-z$ of $\L'$. With a small abuse of notation, from now on  we denote by $A_*$ the
subset of $\O_\L$ given by $\{\s \in\O_\L: \ \s_{\L'}\in A_*\}$. To $A_*$ we can associate two
inner boundaries, $\partial^\L A_*$ and $\partial^{\L'} A_*$, as follows:
\begin{align}
  \label{eq:2}
\partial^\L A_* &= \left\{\s \in A_*: \exists\ x \in \L' \text{ with } c^\L_x
    (\s)=1 \text{ and } \s^x \not \in A_*\right\}\\
\partial^{\L'} A_* &= \left\{\s\in A_*: \exists\ x \in \L' \text{ with } c^{\L'}_x
    (\s)=1 \text{ and } \s^x \not \in A_*\right\}
\end{align}
Clearly $\partial^{\L} A_* \subset \partial^{\L'} A_*$ because
$c_x^\L\le c_x^{\L'}$. Moreover (see Remark \ref{boundaryA}) $\eta\in A_*$ since $\h_x = 1$ for $x
\in [z+1,L-1]$. Thus, if  $\eta_L(t)=1$ then necessarily $\eta(s)\in \partial^\L
A_*$ at some intermediate time $s\le t$. In conclusion
\[
\bbP_\eta^\L\bigl(\eta_L(t)=0\bigr)\ge 1-\bbP_\eta^\L\bigl(\exists
s<t:\ \eta(s)\in \partial^\L A_*\bigr)
\]
We first bound from above $\bbP_\eta^\L\bigl(\eta(s)\in \partial^\L A_*\bigr)$ using the following observation. If the
restriction $\eta_{\L'}$ to $\L'$ was distributed according to
the stationary measure $\pi$ rather then being identically equal to $\mathds{1}{0}$, then this property would be preserved at
any later time. To prove it it is enough to observe that between any two updates of the site $z$ the dynamics in $\L'$ is reversible w.r.t. $\pi$ irrespectively of the actual value of the spin at $z$ and that the updates at $z$ do not depend on the configuration in $\L'$. Therefore
\begin{align*}
\bbP_\eta^\L\bigl(\eta(s)\in \partial^\L A_*\bigr) &\le 
\frac{1}{\pi(\eta_{\L'})}\sum_{\sigma: \ \s_{\L\setminus \L'}=\eta_{\L\setminus\L'}}\pi(\s_{\L'})
\bbP_\s^\L\bigl(\sigma(s)\in \partial^\L A_*\bigr)\\
&\le \frac{e}{q} \pi\bigl(\partial^\L A_*\bigr).
 \end{align*}
Corollary \ref{boundarymeasure} now implies that  
\[
\pi\bigl(\partial^\L A_*\bigr)\le \pi\bigl(\partial^{\L'} A_*\bigr)\le
\frac{q^n 2^{\binom{n}{2} }}{n!}  q^{-(1+\a)}
\]
where $n=\lceil \log_2 L\rceil$. Thus
\[
\bbP_\eta^\L\bigl(\eta(s)\in \partial^\L A_*\bigr)\le e \frac{q^n
  2^{\binom{n}{2} }}{n!}  q^{-(2+\a)}.
\]
In conclusion, a simple union bound over all possible rings in $\L$
within time $t$ (see e.g. \cite{East-Rassegna}*{after (5.12)}) gives
\begin{gather*}
  \bbP_\eta^\L\bigl(\exists s<t:\ \eta(s)\in \partial^\L A_*\bigr) \le
  e L t\frac{q^n 2^{\binom{n}{2} }}{n!}  q^{-(2+\a)}+e^{-L t}=o(1)
\end{gather*}
for all $\epsilon$ small enough. The last identity follows from Theorem \ref{bubu} (ii) and the fact that $t=T_{\rm rel}(\e / q^\g)$.
\end{proof}

\appendix 
\section{Capacity methods}
\label{app:capacity}

In this section we summarize some known results on potential theory
for reversible Markov process, which can be found for example in
\cites{G09,Bovier,BovierNew,ProbOnTreesNets}, in the context of the East
process. In Appendix \ref{capacitino}.1 we give a more detailed motivation for the construction
used to prove the lower bound of Theorem \ref{paletti} and we provide an alternative
proof of the upper bound in Appendix \ref{alt-upper}.2.

We recall the definition of the \emph{electrical
  network} associated to the interval  $\L=[1,L]$.
We consider the undirected graph $\cG_\L$  with vertex set $\O_\L:= \{0,1\}^\L$
and with edges given by $\{\s, \s^x\}$ with $\s \in \O_\L$, $x \in \L$ and
$c^\L_x(\s)=1$. 
That is, there is an edge between two states if and only if there exists a transition between them under the East dynamics.
We denote the edge set by $E_\L$.
Since the East process is reversible we may associate with each edge
$\{\s, \xi\} \in E_\L$ a conductance $c(\s,\xi) = c(\xi,\s)$, generating a weighted graph (or network) in the usual way,
\begin{align*}
  c(\s,\xi):=\p(\s) K(\s,\xi)=\p(\s)c^\L_x(\s)[ p(1-\s_x)+q\s_x ] \,, \quad \text{if }   \xi=\s^x\,, \; x \in \L.
\end{align*}
Equivalently, the resistance is defined as the reciprocal of the conductance $r(\s,\xi)= 1/c(\s,\xi)$.
Note that if $(\sigma,\xi)\notin E_\L$ then the conductance and
resistance are defined as zero and $+\infty$ respectively.
The definition is well posed since $c(\s,\xi)= c(\xi,\s) \geq 0$.

With the above notation the generator of the East process \eqref{thegenerator} can be expressed as
\begin{align*}
  \cL_\L f(\sigma)=\sum_{x\in \L} \frac{c(\sigma,\sigma^x)}{\p(\s)}\left[f(\sigma^x)-f(\sigma)\right]\,.
\end{align*}

Given $B \subset \O_\L$ we  denote by $\t_B$ the hitting time of the set $B$ for the East process $\eta(t)$:
\begin{align*}
\t_B=\inf\{ t > 0 \sst \eta(t) \in B\}\,,
\end{align*}
and denote by $\t_B^+$ the first return time to $B$:
\begin{align*}
\t_B^+=\inf\{ t > 0 \sst \eta(t) \in B,\ \h(s)\neq\h(0)\ \  \textrm{for some}\ \ 0<s<t\}\,.
\end{align*}
We denote by $C_{A,B}$ the \emph{capacity} between two disjoint subsets $A$, $B$ of $\O_\L$ given by (see for example \cite{Beltran} or (3.6) in \cite{BovierNew}):
\begin{align}
\label{eq:cap-def}
 C_{A,B} = \sum_{a\in A}\pi(a)\cR(a) \bbP^\L_a\left(\tau_A^+ > \tau_B\right)\,,
\end{align}
where $\cR(a) = \sum_{\s\neq a}K(a,\s)$ is the holding rate of state $a$.
With slight abuse of notation we write $C_{a,B}$ if $a\not\in B$ is a singleton.
The mean hitting time of $B$ for the East process starting from $a\in \O_\L$ can be expressed in the following way (see for example formula (3.22) in \cite{BovierNew}):
\begin{align}
\label{eq:Thit-Cap}
   \bbE^\L_{a}[\t_B]= \frac{1}{C_{a,B}}\sum_{\sigma \notin B} \p(\sigma)  \bbP^\L_\sigma(\tau_a < \tau_B)\,.
\end{align}
The capacity can also be characterized in terms of variational principles, which are useful for making estimates of $C_{A,B}$.
The following variation principle, useful for finding upper bounds on the capacity, is known as the \emph{Dirichlet principle} (see (3.12) in \cite{G09} or Theorem 3.2 in \cite{Bovier}):
\begin{align}
  \label{eq:dir-princ}
 C_{A,B}&= \inf \{ \cD_\L(f) \sst f: \O_\L\to \bbR, \; f|_A= 1\,,\;f|_{B}=0\}\ ,
\end{align}
where the Dirichlet form $\cD_\L(f)$ is given in \eqref{eq:DirForm}.

For an alternative proof of the upper bound in Theorem \ref{paletti} using a capacity argument we introduce the following definitions and results which can be found in \cite{Levin2008} and \cite{G09}.
We consider the same capacity network described above,
$\cG_\L=(\O_\L,E_\L)$, only now to each edge $\{\s,\h\}\in E_\L$ we
associate two oriented edges $(\s,\h)$ and $(\h,\s)$ (the set of oriented
edges will be written $\tilde E_\L$). For any real valued function $\theta$ on
\emph{oriented edges} we define the divergence at a point $\s\in\O_\L$ by
\begin{align*}
  \div \theta (\s) = \sum_{\h \,:\, \h \sim \s}\theta(\s,\h)\,,
\end{align*}
where $\h \sim \s$ if and only if there exists an edge between them in $E_\L$.
\begin{definition}[Flow from $A$ to $B$]
A flow from the set $A\subset \Omega_\L$ to a disjoint set $B\subset \Omega_\L$, is a real
valued function $\theta$ on $\tilde E_\L$ that is  antisymmetric (\ie
$\theta(\s,\h)=-\theta(\h,\s)$) and satisfies,
\begin{align*}
  \div \theta (\s) &= 0 \quad \textrm{if } \s\notin A\cup B\,,\\
  \div \theta (\s) &\geq 0 \quad \textrm{if } \s\in A\,,\\
  \div \theta (\s) &\leq 0 \quad \textrm{if } \s\in B\,.
\end{align*}
The strength of the flow is defined as  $|\theta| = \sum_{a\in A}\div \theta(a)$.  A flow of strength $1$ is called a \emph{unit flow}. 
\end{definition}
\begin{definition}[The energy of a flow] The energy associated with a
  flow $\theta$ is given by
   \begin{align}
     \label{eq:flowEnergy}
    \cE(\theta) = \sum_{e\in E_\L} r(e) \theta(e)^2\,.
  \end{align}
\end{definition}
\begin{remark}
  The sum in $\cE(\theta)$ is over unoriented edges, so each edge $\{\s,\h\}$ is
only considered once in the definition of energy. Although $\theta$ is defined on oriented
edges, it is antisymmetric and hence $\theta(e)^2$ with $e\in E_\L$ is unambiguous.
\end{remark}
With the above notation \emph{Thomson's Principle} holds, which gives a variational principle for the resistance, 
useful for finding lower bounds on the capacity:
\begin{align}
  \label{eq:thom}
  R(A,B):=\frac{1}{C_{A,B}} = \inf\{ \cE(\theta)\,:\, \theta \textrm{ a unit flow from } A\ \rm{to}\ B\}, 
\end{align}
and, for any finite connected graph, the above infimum is attained by a unique minimiser which we call the \emph{equilibrium flow}.
\subsection*{A.1 Motivation for the proof of the lower bound in Theorem \ref{paletti}}
\label{capacitino}
We now use the above tools to justify our choice of the test function
$\1_{A_*}$ in Section \ref{Astar}.
It turns out that on the mesoscopic scale, $L = d/q^\gamma$, the hitting time $T_\text{hit}(L)$ (see Equation \eqref{raggiungo}) is equivalent, up to constants, to $q$ times $R({\10,B})$ where $B = \{\eta_L = 1 \}$.
So to estimate $T_{\text{hit}}(L)$ it is sufficient to find bounds on
the capacity $C_{\10,B}$.
This is the content of the following lemma.
\begin{lemma}
\label{cicerchie}
Suppose that $L=1/q^{\gamma}$ with $ \gamma \in (0,1]$ and $q<1/2$.
Then there a universal constants $c>0$ such that
\begin{equation}
\frac{qc}{C_{\mathds{1}0,B} } \leq T_{\text{hit}}(L)\leq
\frac{q}{C_{\mathds{1}0,B} }\,, \quad \textrm{where} \quad B = \{\h_L
= 1\}\,.
\end{equation}
\end{lemma}

\begin{proof} 
If $B = \{\eta_L = 1 \}$ then $\bbE_{\1 0}[\tau_B]=T_\text{hit}(L) $.
Since $q < 1/2$ there exists a positive $c := (1/2)^{1/2^\gamma} \leq (1-q)^{1/q^\gamma}$.
We observe that
\begin{align*}
c q \leq q(1-q)^{L-1} = \p(\mathds{1}0) \leq  \sum _{\s \not \in B} \pi(\s) \bbP^\L_\s (\t_{\mathds{1}0}<\t_B) \leq  \pi(B^c)=q\ ,
\end{align*}
the result follows from (\ref{eq:Thit-Cap}).
\end{proof}

We can use Lemma \ref{cicerchie} and the Dirichlet principle \eqref{eq:dir-princ} to get lower bounds on $T_{\text{hit}}(L)$. 
Indeed, for each $f: \O_\L\to
\bbR$ such that  $ f(\mathds{1}0 )= 1$ and $f|_{B}=0$ we have
\begin{equation}\label{fluido}
T_{\text{hit}}(L) \geq \frac{cq}{C_{\mathds{1}0,B} } \geq
\frac{cq}{ \cD_\L(f) }\,.
\end{equation}
It is known that the function $f$ that realizes the minimum in \eqref{eq:dir-princ} is
\begin{equation}\label{pietra}
 f(\eta):=\bbP_\eta^\L ( \tau_{\mathds{1}0}< \tau _B) \, ,
\end{equation}
and so we shall choose a test function for which it is possible to
bound from above the Dirichlet form, and is in someway `close' to $\bbP_\eta^\L( \tau_{\mathds{1}0 }< \tau _B)$.
This motivates the choice of the deterministic dynamics in Section \ref{lwb}.

\subsection*{A.2 An alternative proof of the upper bound in Theorem \ref{paletti}}
\label{alt-upper}
We now give an alternative proof of the upper bound in Theorem \ref{paletti}
using a recursive argument, with the same inductive scheme as used for
the block dynamics proof (see Section \ref{sec:upper-bound}), applied
to flows on the electrical network.

\begin{figure}[t]
  \centering
  \mbox{\hbox{\includegraphics[width=0.75\textwidth]{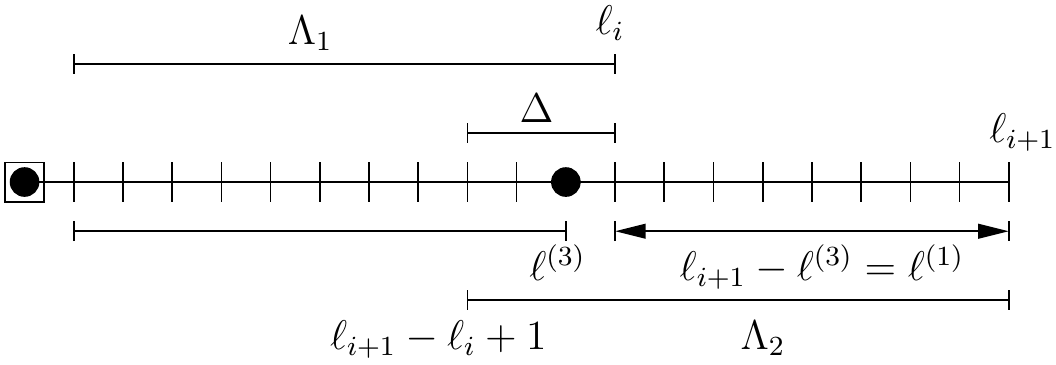}}}
  \caption{\label{fig:lattice} A sketch of the lattice division for one step in the 
    inductive scheme used in the proof of
    Proposition \ref{prop:res-it}. We consider a sequence of
    increasing lengths  $\{ \ell_i\}_{i=1}^r$ defined
    in \eqref{fiesole}. For each $i < r$, $[1,\ell_{i+1}]$ is divided into
    two overlapping intervals $\L_1$ and $\L_2$ of size $\ell_i$, as shown above, with
    intersection $\D$ and $N_{i+1} = |\D|$. The sites in $\D$ are
    parameterised by $\{\ell^{(j)}\}_{j=1}^{N_{i+1}}$.}
\end{figure}

Firstly we recall some notation from the inductive scheme used in
Section \ref{sec:upper-bound}. We consider a sequence $\{ \ell_i\}_{i=1}^r$ of
increasing lengths satisfying \eqref{fiesole}, and let $N_{i+1} =
\lceil \ell_{i} /r \rceil$ (see Fig. (\ref{fig:lattice})).

\begin{proposition}
  \label{prop:res-it}
  Let $L'=\ell_r$, consider the electrical network associated
  with $[1,L']$ and let $R_i$ be the resistance $R(\1,B_{\ell_i})$ where
  $B_{\ell_i}:=\{\eta : \eta_{\ell_i}=0,\ \eta_x = 1\ \textrm{for}\ x
  > \ell_i\}$.
  Then
  \begin{align*}
    R_{i+1} \leq 4R_{i} + \frac{6}{qN_{i+1}}R_{i}\,, \quad \quad \forall\ i < r\,.
  \end{align*}
\end{proposition}

The proof of the upper bound now follows as a corollary from this
proposition together with the results on the inductive scheme
contained in Section \ref{sec:upper-bound}. 
We fix $L\leq d/q$ and 
choose $r = r_0$ as defined at the end of Section
\ref{sec:upper-bound}. Since (recall $\ell_{r_0}\le 2^{r_0+1}$)
$$q N_{i+1} =
q\lceil \ell_i / r_0 \rceil \leq q\ell_i + 1\leq  q \ell_{r_0}+1\leq  8d/c_0+1$$
we have $4R_i \leq (c(d) - 6) R_i / (q
N_{i+1})$ for some positive constant $c(d)$ depending only on $d$. 
The proposition above therefore implies that
\begin{align*}
  R_{i+1} \leq  \frac{c(d)}{qN_{i+1}}R_{i}\,, \quad \quad \forall i < r_0 \,,
\end{align*}
hence
\begin{align*}
  R_r \leq R_1 \frac{c(d)^rr^r}{q^r\prod_{i=1}^{r-1}\ell_i}\,, \ \quad
  \textrm{where } r = r_0\,.
\end{align*}
Comparing with \eqref{ringhiera1} and using the arguments at the end
of Section \ref{sec:upper-bound}  this gives rise
to
\begin{align}
\label{eq:7}
 R(\1,B_{\ell_{r_0}}) = R_{r_0} \leq q^{-c'(d)} \frac{n!}{q^n 2^{\binom{n}{2}}}\,.
\end{align}
Recall that $\ell_{r_0} \geq L$. We observe from \eqref{eq:cap-def}  that
\begin{align}
  \label{eq:6}
  R(\1,B_L) \leq R(\1,B_{\ell_{r_0}})\,,
\end{align}
since starting from $\1$ the East dynamics must cross $B_L$ to reach
$B_{\ell_{r_0}}$.
The same proof as for Lemma \ref{cicerchie} shows that
$\bbE_{\1}^\L[\t_{\h_L=0}] \leq R(\1,B_L)$, and by Proposition
\ref{dominare} we have $T_{\rm hit}(L) \leq
\bbE_{\1}^\L[\t_{\h_L=0}]$, so the upper
bound on the relaxation time in Theorem \ref{paletti} follows as a
consequence of the equivalence of the characteristic times in Theorem
\ref{equo} together with \eqref{eq:7} and \eqref{eq:6}.

\begin{figure}[t]
  \centering
  \mbox{\subfloat{\vtop{%
  \vskip0pt
  \hbox{\includegraphics[width=0.48\textwidth]{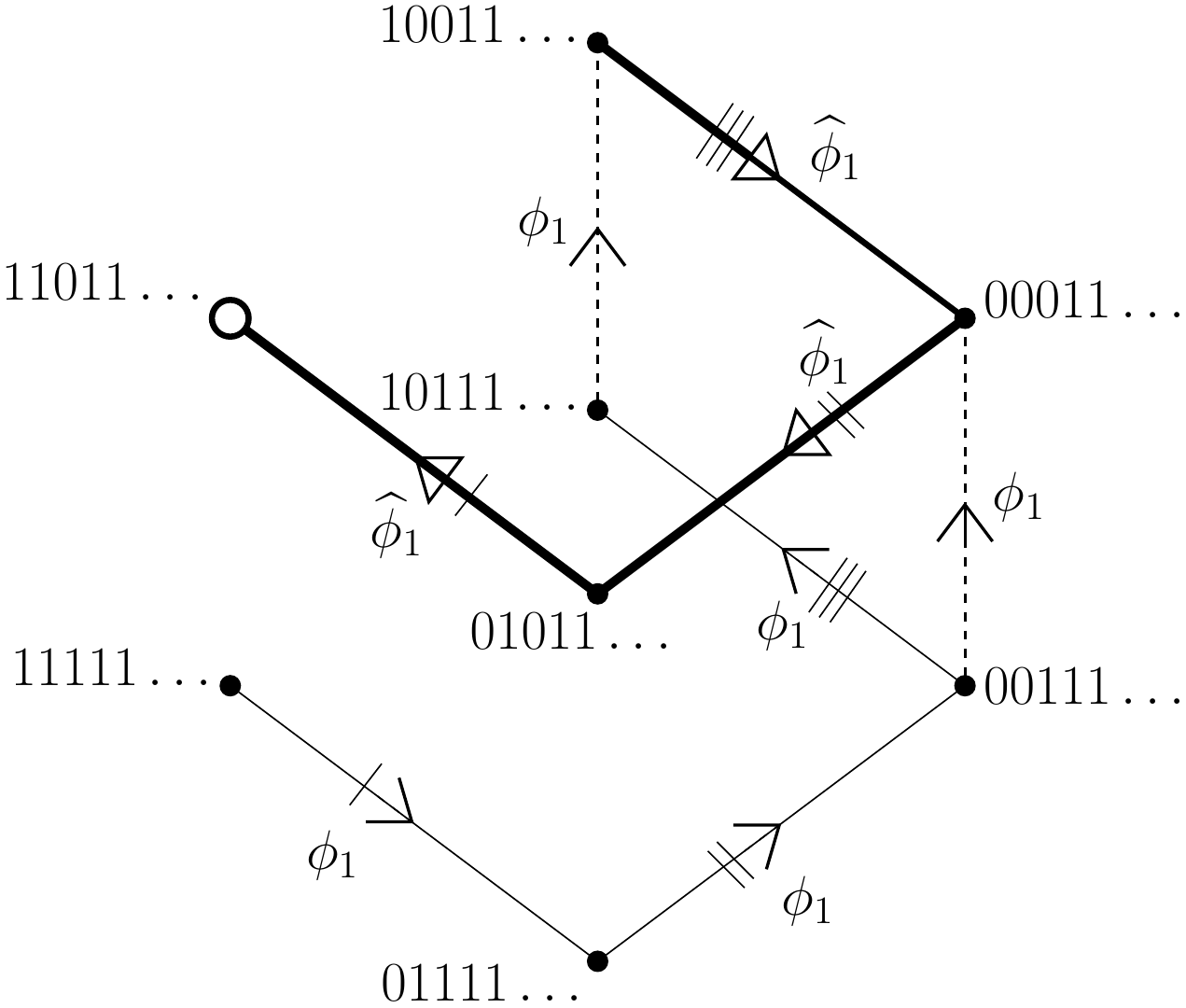}}}}\quad \quad 
    \subfloat{\vtop{%
  \vskip0pt
  \hbox{\includegraphics[width=0.48\textwidth]{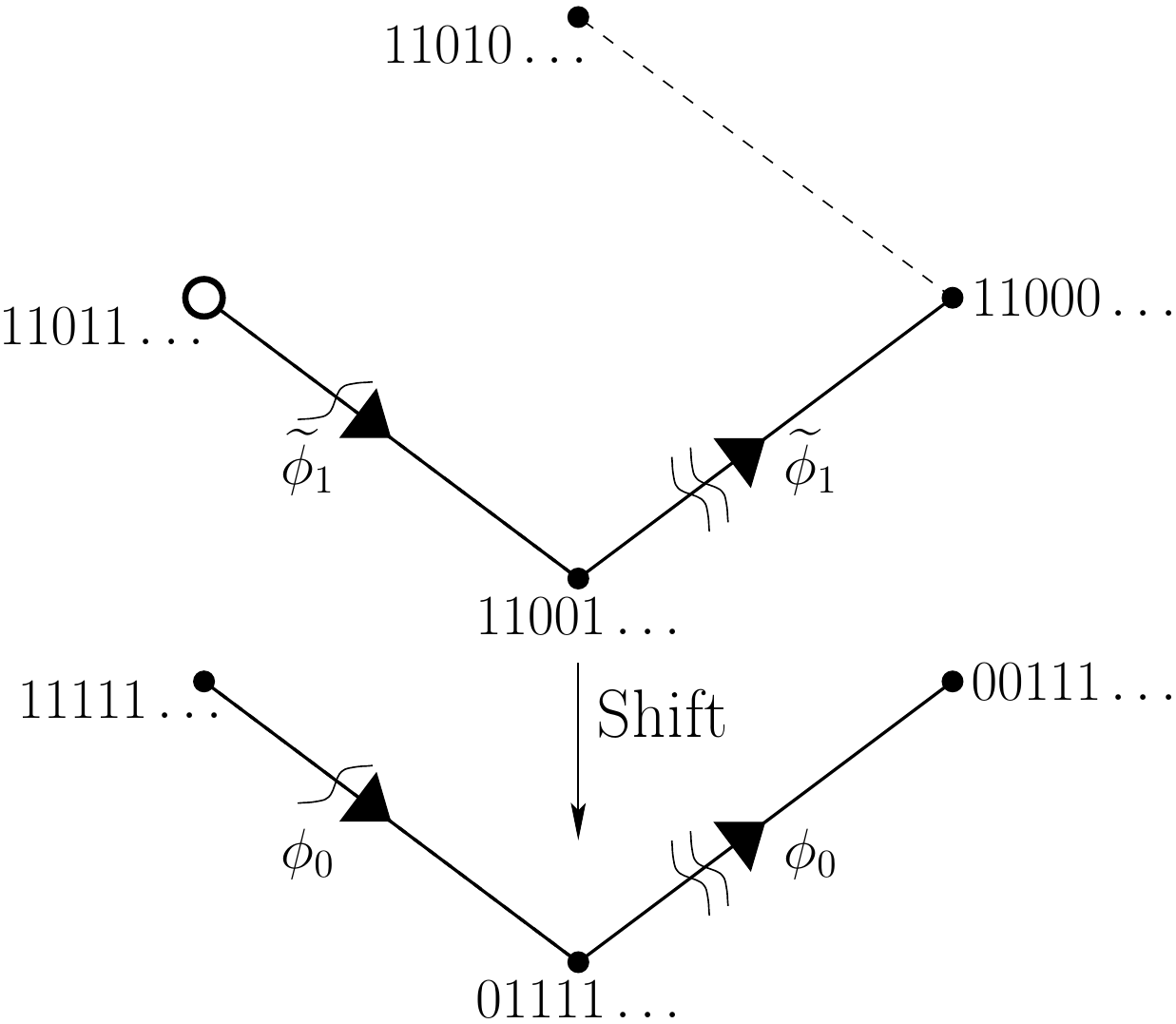}}}}}
  \caption{\label{fig:cube} The 
    construction of the flow $\phi_1+\widehat\phi_1+\widetilde\phi_1$ in the first step ($i = 1$) of the inductive scheme used in the proof of
    Prop. \ref{prop:res-it}. Note $\ell_1 = 3$, $\ell_2 = 5$, $\D =
    \{3\}$  and $N=1$, for all $r>2$ (see \eqref{fiesole}). 
    $\phi_1$, $\widehat\phi_1$ and $\widetilde\phi_1$ have disjoint
    support.
     Arrows show the direction of the flow and the labels indicate
     which flow is non-zero on each edge.
     The flow $\widehat\phi_1$ is obtained from $\phi_1$ by a projection
     and inversion, the ticks  indicate edges with equal flow
     strength. 
     $\widetilde\phi_1$ is obtained via association with $\phi_0$
     under a shift. The left and right images join on the common
     vertex shown by the `$\circ$'.
}
\end{figure}

\begin{proof}[Proof of Proposition \ref{prop:res-it}]
Fix $i < r$, similarly to Section \ref{sec:upper-bound} we consider the interval
$\L^{(i+1)}=[1,\ell_{i+1}]$ dived into two overlapping intervals
\begin{align*}
  \L_1 := [1,\ell_i] \quad \L_2 := [\ell_{i+1} - \ell_i + 1, \ell_{i+1}].
\end{align*}
The intersection $\D := \L_1 \cap
\L_2 = [\ell_{i+1} - \ell_i +1, \ell_{i}]$ contains $N_{i+1}\geq 1$
sites by \eqref{fiesole} (see Fig. \ref{fig:lattice}). 
To reduce notation throughout the proof we fix $N:=N_{i+1}$.

Let $\ell^{(j)} = (\ell_i - N + j) \in \D $ for $j\in\{1,\ldots,N\}$ and
define $\ell^{(0)} := \ell_i - N$.
For $0\leq j \leq N$ let $\phi_j$ be the equilibrium unit flow from
$\1$ to $B_{\ell^{(j)}}$.
By Thomson's Principle and the same argument leading to \eqref{eq:6} we have
\begin{align*}
  \max_{j\in\{0,1,\ldots,N\}} \cE(\phi_j) = \cE(\phi_N) = R_{i}\,.
\end{align*}

Fix $j\in\{1,\ldots,N\}$. 
We now define $\widehat\phi_j$, an antisymmetric function on oriented
edges, with support on edges between elements of $B_{\ell^{(j)}}$,
such that $\phi_j+\widehat\phi_j$ is a unit flow from $\1$
to the configuration $\10_j\1:=\{\h: \eta_{\ell^{(j)}} = 0,\ \eta_{x} =
1,\textrm{ for } x \neq \ell^{(j)}  \}$.
The East dynamics on the first $\ell^{(j)}-1$ sites is not influenced
by the spin on site $\ell^{(j)}$, so the structure of the electrical
network between configurations in $B_{\ell^{(j)}}$ is identical to
that on $\{\h : \h_x = 1 \textrm{ for } x \geq \ell^{(j)}\}$
 up to a factor of $p/q$ in the edge resistance (see Fig. \ref{fig:cube}).
We therefore define $\widehat\phi_j$ by `reversing'  $\phi_j$ on edges
which are equivalent under a projection onto the first $\ell^{(j)}-1$ sites,
\begin{align}
  \label{eq:phihat}
  \widehat \phi_j(\s,\h) :=
  \begin{cases}
    \phi_j(\h^{\ell^{(j)}},\s^{\ell^{(j)}})& \textrm{ if } \s,\h\in B_{\ell^{(j)}}\,,
    \\
    0 & \textrm{otherwise}\,.
  \end{cases}
\end{align}
Recall from \eqref{eq:flip} that $\h^{\ell^{(j)}}$ is the configuration $\h$  with the
spin at site $\ell^{(j)}$ flipped (in this case flipped to $1$). 
It is straightforward to check that $\phi_j+\widehat \phi_j$ defines a unit flow from $\1$ to the point $\10_j\1$, 
we postpone the proof until the end.

We now define $\widetilde\phi_j$ as a unit flow from $\10_j\1$ to $B_{\ell_{i+1}}$. 
Observe that $\ell_{i+1}-\ell^{(j)} = \ell^{(N-j)} < \ell_i$, see Fig. \ref{fig:lattice}. 
So we define $\widetilde\phi_j$ by keeping the vacancy at $\ell^{(j)}$
fixed and `shifting' the equilibrium flow from $\1$ to
$B_{\ell^{(N-j)}}$ (given by $\phi_{N-j}$) onto the lattice
$[\ell^{(j)}+1,\ell_{i+1}]$, where the constraint on the site $\ell^{(j)}+1$ is always satisfied because of the fixed vacancy.
Define $C_{\ell^{(j)}} :=\{\h : \h_{\ell^{(j)}} = 0 \textrm{ and }
\eta_x=1 \textrm{ for } x < \ell^{(j)} \}$ then
\begin{align*}
  \widetilde\phi_j(\s,\h) =
  \begin{cases}
    \phi_{N-j}(\widetilde{\s},\widetilde{\h}) & \textrm{if } \s,\h \in C_{\ell^{(j)}}\,, \\
    0 & \textrm{otherwise,}
  \end{cases}
\end{align*}
where the shift is given by
\begin{align*}
\widetilde{\h}_x =
  \begin{cases}
    \h_{x+\ell^{(j)}} & \textrm{if} \ x \leq \ell^{(N-j)}\,,\\
    1 & \textrm{otherwise.} \
  \end{cases}\\
\end{align*}

\begin{claim}
  \label{claim:unitFlow}
  For each $j\in \{1,\ldots,N\}$,  $\phi_j +
  \widehat\phi_j + \widetilde\phi_j$ is a unit flow from $\1$ to $B_{\ell_{i+1}}$.
\end{claim}
In light of this claim, now define $\Theta$ as the normalised sum of the unit flows from $\1$ to
$B_{\ell_{i+1}}$ over $j\in\{1,\ldots,N\}$;
\begin{align*}
  \Theta = \frac{1}{N}\sum_{j=1}^N& \left(\phi_j +
  \widehat\phi_j + \widetilde\phi_j\right)\ \quad \textrm{and} \\
\Phi = \frac{1}{N}\sum_{j=1}^N \phi_j \,,\quad \widehat\Phi&= \frac{1}{N}\sum_{j=1}^N \widehat\phi_j \,, \quad\widetilde\Phi =\frac{1}{N}\sum_{j=1}^N \widetilde\phi_j\,.
\end{align*}
Since $B_{\ell^{(i)}}\cap B_{\ell^{(j)}} = \emptyset$, for $i\neq j$, 
$\{\widehat \phi_j\}_{j=1}^N$ have disjoint support, also
$C_{\ell^{(i)}}\cap C_{\ell^{(j)}}=\emptyset$ for $i\neq j$, so the
same holds for $\{\widetilde\phi_j\}_{j=1}^N$, therefore by iterating Lemma
\ref{lemma:1} (3) we have $\cE(\widehat\Phi) = \frac{1}{N^2}
\sum_{j}\cE(\widehat \phi_j) $ (and similarly for $\widetilde\phi_j$).
It follows, again from Lemma \ref{lemma:1}, that
\begin{align}
  \label{eq:total-flow}
  \cE(\Theta)&\leq 4\left( \cE(\Phi) + \cE(\widehat\Phi)\right) +2\cE(\widetilde\Phi) \nonumber\\
  &\leq 4\left(  \max_{j} \cE(\phi_j) + \frac{1}{N} \max_{j} \cE(\widehat \phi_j) \right) + \frac{2}{N} \max_j \cE(\widetilde\phi_j)\,.
\end{align}
Also for each $(\s,\h)$ with $\widehat \phi_j(\s,\h)>0$ there exists a unique edge $(\s^{\ell^{(j)}},\h^{\ell^{(j)}})$ such that $\phi_j(\s^{\ell^{(j)}},\h^{\ell^{(j)}})>0$ and $r(\s,\h)= p\, r(\s^{\ell^{(j)}},\h^{\ell^{(j)}})/q$ (similarly for $\widetilde\phi_j$), so
\begin{align*}
  \max_j\cE(\widehat\phi_j) &\leq \frac{p}{q}\max_j \cE(\phi_j)  \leq \frac{\cE(\phi_N)}{q} = \frac{R_i}{q}\,,\\
  \max_j\cE(\widetilde\phi_j) &\leq \frac{p}{q}\max_j \cE(\phi_j) \leq \frac{\cE(\phi_N)}{q}= \frac{R_i}{q}\,.
\end{align*}
The result now follows by combining the above bounds with \eqref{eq:total-flow} and
applying Thomson's Principle \eqref{eq:thom}, since $\Theta$ is a unit
flow from $\1$ to $B_{\ell_{i+1}}$ (combine Claim \ref{claim:unitFlow}
with Lemma \ref{lemma:1} part (1)).
\end{proof}

\begin{proof}[Proof of Claim \ref{claim:unitFlow}]
 Fix $j\in\{1,\ldots,N\}$, we show that $\theta = \phi_j + \widehat\phi_j + \widetilde\phi_j$ is a unit flow from $\1$ to $B_{\ell_{i+1}}$.
 Firstly observe that  $\phi_j$, $\widehat\phi_j$ and
 $\widetilde\phi_j$ have support on three 
 disjoint edge sets;
 \begin{align*}
   \begin{cases}
     \phi_j(\s,\h) > 0 &  \Rightarrow \quad (\s,\h) \in  E_1\,,\\
     \widehat\phi_j(\s,\h) > 0 & \Rightarrow \quad (\s,\h) \in E_2\,,\\ 
     \widetilde\phi_j(\s,\h) > 0 & \Rightarrow \quad (\s,\h) \in E_3\,,
   \end{cases}
 \end{align*}
 where, setting $\L=[1,\ell_r]$,
 \begin{align*}
      E_1 &= \{(\s,\h)\in \tilde E_\L : \s_x{=}\h_x{=}1\, \forall x>\ell^{(j)},\textrm{ and at most one of } \s,\h \textrm{ have a vacancy at } \ell^{(j)} \}\,,\\
      E_2 &= \{(\s,\h)\in \tilde E_\L : \s_x{=}\h_x{=}1\, \forall x>\ell^{(j)},\ \s_{\ell^{(j)}}{=}\h_{\ell^{(j)}} {=} 0\} \  \textrm{and finally,}\\
      E_3 &= \{(\s,\h)\in \tilde E_\L : \s_x{=}\h_x{=}1\, \forall
      x<\ell^{(j)},\ \s_{\ell^{(j)}}{=}\h_{\ell^{(j)}} {=} 0\}\,.
 \end{align*}
 $\div \theta(\1) = 1$ since there exists only a single edge connected
 to the state $\1$, and this edge belongs to $E_1$,
 then since $\phi_j$ is a unit flow from $\{\1\}$ we must have $\div \phi_j(\1) = 1$.
 We now check that $\div \theta(\s) = 0$ for $\s \in B_{\ell_{i+1}}^c\setminus \{\1\}$.
 If $\s \in \{\h : \h_x = 1 \,\forall x \geq \ell^{(j)}\}$, then $\div \theta(\s) = \div \phi_j(\s) = 0$.
 Now fix $\s \in B_{\ell^{(j)}}\setminus\{\10_j\1\}$, so that
 $\s_{\ell^{(j)}}=0$, $\s_x =1$ for $x>\ell^{(j)}$ and $\s_y=0$ for some $y<\ell^{(j)}$, then
 $\theta(\s,\s^x)>0$ implies $(\s,\s^x)\in E_1\cup E_2$.
 In particular $\theta(\s,\s^x)>0$ implies $x \leq \ell^{(j)}$, so
 \begin{align*}
   \div \theta(\s) &=  \sum_{\substack{x\leq\ell^{(j)} : \\\,
       \s_{x-1}=0}}\theta(\s,\s^x) =
   \phi_j(\s,\s^{\ell^{(j)}})\1_{\{\s_{\ell^{(j)}-1} = 0\}}(\s) +
   \sum_{\substack{x <\ell^{(j)} :\\\s_{x-1}=0}}\widehat\phi_j(\s,\s^x)  \\
   &= 
   -\phi_j(\s^{\ell^{(j)}},\s)\1_{\{\s_{\ell^{(j)}-1} = 0\}}(\s) -
   \sum_{\substack{x <\ell^{(j)} : \\\s_{x-1}=0}}\phi_j(\s^{\ell^{(j)}},(\s^x)^{\ell^{(j)}})\\
   &= -\div \phi_j\left(\s^{\ell^{(j)}}\right) = 0\,.
 \end{align*}
  Finally it is simple to check directly that the divergence on the
 configuration $\10_j\1$ is zero since there are only two
 configurations which are reachable from here under the East dynamics,
 by flipping the spin on site $1$ or on site $\ell^{(j)}+1$ 
 \begin{align*}
   \div \Theta(\10_j\1) &= \widetilde\phi_j(\10_j\1,\10_j\1^{\ell^{(j)}+1}) -
   \widehat\phi_j(\10_j\1^1,\10_j\1) \\
   &= \phi_{N-j}(\1,\1^1)- \phi_j(\1,\1^1)=1-1=0\,.
 \end{align*}
 The remaining relevant configurations are given by $\{\h : \h_x = 1 \textrm{
   for } x < \ell^{(j)},\ \h_{\ell^{(j)}} = 0\}$, but on this set the
 flow is simply given by the unit flow from $\1$ to $\{\h : \eta_{\ell_{i+1}} =
 0\}$ on the lattice $[\ell^{(j)}+1,\ell_{i+1}]$ with a zero boundary
 condition at $\ell^{(j)}$ and therefore zero divergence is inherited
 from $\phi_{N-j}$. Non-positive divergence on $B_{\ell_{i+1}}$ is
 also inherited from $\phi_{N-j}$.
\end{proof}

\begin{lemma}
  The following three results are used in the proof of the upper bound using flows:
  \label{lemma:1}
  \begin{enumerate}
  \item If $\{\theta_i\}_{i=1}^N$ are unit flows from $A$ to $B_i$ then $\frac{1}{N}\sum_{i=1}^N \theta_i$ is a unit flow from $A$ to $\bigcup_i B_i$ and
  \begin{align}
    \label{eq:EnergyCS}
    \cE\left(\frac{1}{N}\sum_{i=1}^N \theta_i \right) \leq  \max_{i\in\{1,\ldots,N\}} \cE(\theta_i)\ .
  \end{align}
  \item   For two flows $\theta_1$ and $\theta_2$,
    \begin{align}
      \label{eq:CS2}
      \cE(\theta_1 + \theta_2) \leq 2\left(\cE(\theta_1) + \cE(\theta_2) \right)\ .
    \end{align}
  \item   Suppose $\Theta = \theta_1 + \theta_2$ is a flow from $A$ to $B$ and $\theta_1(e) \neq 0$ implies $\theta_2(e) = 0$. Then
  \begin{align}
    \label{eq:indepFlow}
    \cE(\Theta) = \cE(\theta_1) + \cE(\theta_2)\,.
  \end{align}
  \end{enumerate}
\end{lemma}
\begin{proof}
  Let $\Theta(x,y) := \frac{1}{N}\sum_{i=1}^N \theta_i(x,y)$,
  since each $\theta_i$ is antisymmetric and a linear combination of antisymmetric functions is antisymmetric so is $\Theta$.
  Zero divergence on $(A\cup B)^c$, non-negative divergence on $A$ and
  non-positive on $B$, and unit strength all follow from linearity of the divergence.
  So $\Theta$ is a unit flow from $A$ to $\bigcup_i B_i$.
  Inequality (\ref{eq:EnergyCS}) and \eqref{eq:CS2} both follow from
  simple applications of the Cauchy-Schwarz inequality.

  Part (3) is immediate from the definition of the energy, by decomposing the sum in \eqref{eq:flowEnergy} over two non intersecting sets, 
  one on which $\theta_1$ is non-zero and another on which $\theta_2$ is non-zero.
\end{proof}




\end{document}